\documentclass[10pt,journal,compsoc]{IEEEtran}
\usepackage{cite}
\usepackage{amsmath,amssymb,amsfonts}
\usepackage{tabularx,booktabs}
\usepackage{CJK}
\usepackage{algorithmic}
\usepackage{graphicx}
\usepackage{textcomp}
\usepackage{xcolor}
\usepackage{threeparttable}
\usepackage{multirow}
\usepackage{subfigure}
\usepackage{threeparttable}
\usepackage{caption}
\usepackage{amsthm}
\usepackage{epstopdf}
\usepackage{enumitem}

\newtheorem{proposition}{Proposition}

\def\BibTeX{{\rm B\kern-.05em{\sc i\kern-.025em b}\kern-.08em
		T\kern-.1667em\lower.7ex\hbox{E}\kern-.125emX}}

\begin{document}
%
\title{A Misreport- and Collusion-Proof Crowdsourcing Mechanism without Quality Verification}
%
%
%
%

\author{Kun Li, Shengling Wang*, ~\IEEEmembership{Member,~IEEE}, Xiuzhen Cheng, Qin Hu	
\IEEEcompsocitemizethanks{		

\IEEEcompsocthanksitem Kun Li and Shengling Wang (Corresponding author) are with School of Artificial Intelligence, Beijing Normal University (BNU), China. E-mail: likun@mail.bnu.edu.cn, wangshengling@bnu.edu.cn.

\IEEEcompsocthanksitem Xiuzhen Cheng is	with School of Computer Science and Technology, Shandong University (SDU), China. E-mail: xzcheng@sdu.edu.cn

\IEEEcompsocthanksitem Qin Hu is with the Department of Computer and Information Science, Indiana University - Purdue University Indianapolis, IN, USA. \protect E-mail: qinhu@iu.edu

}

\thanks{
	This work has been supported by National Key R\&D Program of China (No.2019YFB2102600), National Natural Science Foundation of China (No.61772080) and Engineering Research Center of Intelligent Technology and Educational Application, Ministry of Education.
}
}

\IEEEtitleabstractindextext{%
\begin{abstract}
Quality control plays a critical role in  crowdsourcing. The state-of-the-art work is not suitable  for large-scale crowdsourcing applications, since  it is a long haul for the requestor to verify  task quality or select professional workers in a one-by-one mode.  In this paper, we propose a misreport- and collusion-proof crowdsourcing mechanism,  guiding workers to truthfully report the quality of submitted tasks without collusion by designing a mechanism, so that workers have to act the way the requestor would like. In detail, the mechanism proposed by the requester makes no room for the workers to obtain profit through quality misreport and collusion, and thus, the quality can be controlled without any verification. Extensive simulation results verify the effectiveness of the proposed mechanism. Finally, the importance and originality of our work lie in that it reveals some interesting and even counterintuitive findings: 1) a high-quality worker may pretend to be a low-quality one; 2) the rise of task quality from high-quality workers may not result in the increased utility of the requestor; 3)  the utility of the requestor may not get improved with the increasing number of workers. These findings can boost forward looking and strategic planning solutions for crowdsourcing.
\end{abstract}

\begin{IEEEkeywords}
Crowdsourcing, Collusion-proof, Mechanism design.
\end{IEEEkeywords}}

\maketitle

\IEEEdisplaynontitleabstractindextext

%
\IEEEpeerreviewmaketitle

\section{Introduction}\label{Introduction}
The Internet openness renders crowdsourcing  to gather geographically dispersed human resources for accomplishing complex tasks  that are easy for human beings while difficult for machines. However, just also due to the openness of the Internet, any people can apply for participating in crowdsourcing, which may incur  unprofessional crowdsourcees (workers)  with low-quality contributions. Thus, there is a pressing need for quality control in crowdsourcing. The state-of-the-art quality control in crowdsourcing can be categorized into two kinds: task-based \cite{oleson2011programmatic, wu2017photo, xu2015revealing, jin2017leveraging} and worker-based \cite{qiu2017dynamic, hu2019quality, han2016crowdsourcing, he2015high, tarable2015importance}. The former kind is a direct way which proposes  different approaches to evaluate or develops various  tools to monitor the task quality. The latter kind is an indirect means that selects qualified workers to guarantee the quality of submissions.

Existing quality control is not applicable for large-scale crowdsourcing applications, such as urban traffic monitoring and air-quality sensing,  since  it is a long haul for the crowdsourcer (requestor)  to verify task quality or select professional workers in a one-by-one mode. In this scenario, it is convenient to require workers to report the quality of their tasks based on which the requestor  pays them.  Such a naive way will obviously lead to two issues: 1) {\it misreport}. A worker may  report  his low-quality task as a high-quality one dishonestly. 2) {\it collusion.} High-quality workers\footnote {A high-quality worker refers to the one who submits the task with high quality. Accordingly, the low-quality worker is the one with the low-quality submitted tasks.} can save cost by recruiting low-quality ones to work for them, through which the low-quality workers can also gain extra income.

The above two issues make it ridiculous for the requestor to pay relying on the quality reported by the workers themselves. However, the nature of requiring no quality verification renders such a way attractive to large-scale crowdsourcing applications. Hence, to make this seemly ridiculous way  feasible, we propose a misreport-  and collusion-proof crowdsourcing without quality verification in this paper. The aim of our scheme is guiding workers to truthfully report the quality of submitted tasks without collusion by leveraging pricing and task allocation. In other words,  driven by the  market power, workers have to try their best to serve the requestor honestly, making the quality of tasks can be naturally guaranteed.

However, it is challenging to realize our aim since the capability  of a worker to complete  the task is kept private to the requestor. This information asymmetry leads to the uncertainty of the requestor's optimal strategy, which in turn causes the obscureness of  that of any worker, resulting in that the Pareto optimality cannot be achieved in crowdsourcing. To tackle the above challenge, we resort to the mechanism design game theory \cite{borgers2015introduction}, which empowers the requestor to dominate the game with workers through designing a {\it mechanism (game rule)}, so that workers have to  act the way  the requestor would like. In detail, the mechanism proposed by the requester  makes no room for the workers to obtain profit through quality misreport and collusion, and thus, the quality  can be controlled without any verification.

To the best of our knowledge, this is the first  work that simultaneously guards against quality misreport and collusion among workers in crowdsourcing. The contributions in our paper are summarized as follows:
\begin{itemize}
	\item A crowdsourcing framework without quality verification is proposed, where the mechanism design game theory is leveraged to guide workers to behave honestly.
	\item A special crowdsourcing mechanism for the two-worker model is designed, which includes three kinds of constraints: the participation-incentive constraint, the incentive compatibility constraint and the collusion-proof constraint. In addition, the property of the  cost function for any worker is deduced.
	\item A general crowdsourcing mechanism for multiple workers is extended from  the basic two-worker model, which has a particularly-designed collusion-proof constraint for the
	optimal collusion scheme in the multiple-worker scenario.
	\item Extensive simulation results verify the effectiveness of the proposed   misreport- and collusion-proof crowdsourcing mechanism.
\end{itemize}

Finally, the importance and originality of our work lie in that it reveals some interesting and even counterintuitive findings,  providing fresh insights into understanding the complexity of crowdsourcing. These  findings can boost forward looking and strategic planning solutions, which  are summarized as follows:
\begin{itemize}
	\item {\it A high-quality worker may pretend to be a low-quality one.}  The reason behind this anti-intuitive fact is when the number of low-quality workers is large, or when the task quality of high-quality workers enhances, the profit of misreport by low-quality workers increases. Thus, the requestor has to use high payment to  drive low-quality workers to behave honestly,  which in turn creates the motivation for high-quality workers to lie.
	\item {\it The rise of task quality from high-quality workers may not result in the increased utility of the requestor.} When the task quality of low-quality workers remains unchanged, the increase of task quality from high-quality workers implies the profits of collusion and misreport go up, leading to more incentives for malicious behaviors. In this case, the requestor has to pay more to reduce the possibility of misreport and collusion, thus may reduce her\footnote{In this paper, we use ``she" and ``he" to indicate the requestor and the worker, respectively.} utility.
	\item {\it The utility of the requestor  may not get improved with the increasing number of workers.}
	Specifically, more workers mean more chances of collusion and misreport, leading to high cost for the requestor to prevent the occurrence of these problems, which may lower her utility.
\end{itemize}

The rest of this paper proceeds as follows. In Section \ref{sec:related}, we summarize the related work on quality control in crowdsourcing. An overview of our proposed quality control framework is presented in Section \ref{sec:overview}, which is specifically elaborated in Section \ref{sec:two} for the basic two-worker model and further extended in Section \ref{sec:m} for the general $m$-worker scenario. In Section \ref{sec:performance}, we conduct substantial simulation experiments to evaluate the performance of our proposed mechanism and reveal several interesting findings. Finally, we conclude the whole paper in Section \ref{sec:conclusion}.

\section{Related work}\label{sec:related}
In recent years, as crowdsourcing is widely used in various fields, quality control has become the focus of research in the crowdsourcing area. Existing research on quality control can be generally classified into two types: \textit{task-based} and \textit{worker-based}.

Task-based quality control in crowdsourcing is to assess the quality of the tasks submitted by workers through various indicators. One of the most common methods is to use gold standard data \cite{oleson2011programmatic} to measure the quality of tasks. By comparing with standard data, it is easy to distinguish unqualified tasks from qualified ones.
In photo crowdsourcing, Wu \emph{et al.} \cite{wu2017photo} found it challenging to use limited resources to get photos covering the target area as much as possible, so they adopted a scheme named image quality assessment (IQA) to address this challenge, which was deployed in mobile devices to filter out those low-quality photos before sending metadata to the server thereby achieving quality control.
And when investigating crowdsourcing-based spammers, Xu \emph{et al.} \cite{xu2015revealing} revealed that the quality control was conducted through strictly setting the standard to check whether the work done by the spammer was qualified to be paid or not, which enforced a low payment rate of  22.6\% so as to stimulate spammers to improve the quality of their spam posts.
However, in most crowdsourcing scenarios, it is impractical or unnecessary to prepare gold standard data in advance, so researchers proposed some machine-learning based methods to assess task quality.
In \cite{jin2017leveraging}, Yuan \emph{et al.} tried to improve data quality in the process of crowdsourcing-based labeled data collection through extending a classic probabilistic model named GLAD \cite{whitehill2009whose}, which systematically encoded different types of side information, such as worker, item and context information.

Worker-based quality control is often achieved by selecting the appropriate set of workers with the help of some well-designed models or algorithms based on the workers' ability or reputation.
In \cite{qiu2017dynamic}, in order to slect high-quality workers, Qiu \emph{et al.} designed a dynamic contract for each worker based on the worker's performance and objective, so as to elicit high-quality submissions and prevent malicious behaviors.
Hu \emph{et al.} \cite{hu2019quality} used an economics-based philosophy to improve workers' quality in crowdsourcing. Taking advantage of their proposed incentive algorithms based on the sequential zero-determinant strategy, the requestor could use the market power to stimulate workers to submit high-quality results.
Han \emph{et al.} \cite{han2016crowdsourcing} presented a new crowdsoucing system to provide annotations for web information extraction, which could collect a wide set of behavioral features and predict annotation quality of each worker for annotating web page structure. They collected a set of workers' behavioral features and discovered the relationship between the crowdsourcing quality and the workers' behavioral features.
In vehicle-based crowdsourcing \cite{he2015high}, the relationship between spatiotemporal coverage and the vehicle trajectory was studied to design a new strategy for worker recruitment, which guaranteed the crowdsoucing quality by employing the best set of participants meeting the application requirements.
In \cite{tarable2015importance}, Tarable \emph{et al.} proposed a ``maximum a-posteriori" decision rule to help the requestor make decisions in reputation-based task allocation for crowdsourcing, which worked well even in the case of inaccurate reputation information.

In summary, all the above quality-control schemes are relying on additional assessing indicators or algorithms to select tasks or workers in a one-by-one manner, which are obviously tedious and inefficient, leading to their inapplicability in large-scale crowdsourcing scenarios. Therefore, we propose a simple but effective quality control model in crowdsourcing, which can further eliminate the collusion behavior at the same time.
\section{Overview of our framework}\label{sec:overview}
According to the analysis in Section \ref{Introduction}, achieving quality control without verification will lead to the problems of {\it quality misreport} and {\it collusion}. Both problems originate from that the worker's capability  to complete tasks is his private information in the absence of quality verification. Without this information, the optimal strategy of the requestor is uncertain, and so is that of any worker. To optimize the incomplete-information game between the requestor and workers, we take advantage of the mechanism design game theory \cite{borgers2015introduction}.

The mechanism design game theory is an efficient vehicle to solve the game with private information, which introduces two kinds of players, namely ``the agent" and ``the principal". The agent  has private information called ``type" while the principle does not have and is unable to access the private information of the agents.   To achieve the Pareto optimality, the mechanism design game theory allows the principal to dominate the game with the agents through designing a mechanism (game rule),  which leads the agents to act the way the principal would like. Specifically, the principal asks the agents to report their types (called \emph{direct mechanism}) or strategies made based on the private information (called \emph{indirect mechanism}).\footnote{According to the revelation principle, for every Bayesian Nash equilibrium,  there exists a Bayesian game with the same equilibrium outcome but in which players truthfully report types, i.e., both the direct mechanism and indirect mechanism produce the same results. In this way, the principal only needs to consider the information provided by the agents to develop the optimal strategy, so as to significantly reduce the complexity.} Then the principal and the agents act in light of the mechanism the principal designed and gain the corresponding payoffs. Whether an agent reports the true type or strategy based on the true private information depends on if the mechanism formulated by the principal satisfies the incentive compatibility (\textbf{IC}) constraint. The \textbf{IC} constraint guarantees that the agents are motivated to behave in a manner consistent with the principal's optimal strategy.

In our scenario, since workers have private information, they are agents whose  types are  their working capabilities. Correspondingly, the requestor is the principal who can design  a mechanism $\mathbb{G}$  with the \textbf{IC} constraint to force any worker to report the task quality based on his working capability truthfully, thus solving the problem of quality misreport. Additionally, the collusion-proof constraint is also included in the proposed mechanism $\mathbb{G}$ to address the issue that multiple malicious workers collude to deceive the requestor. In this paper, we assume that malicious workers are rational and intelligent. 
\begin{figure}[]
	\centering
	\includegraphics[width=0.5\textwidth]{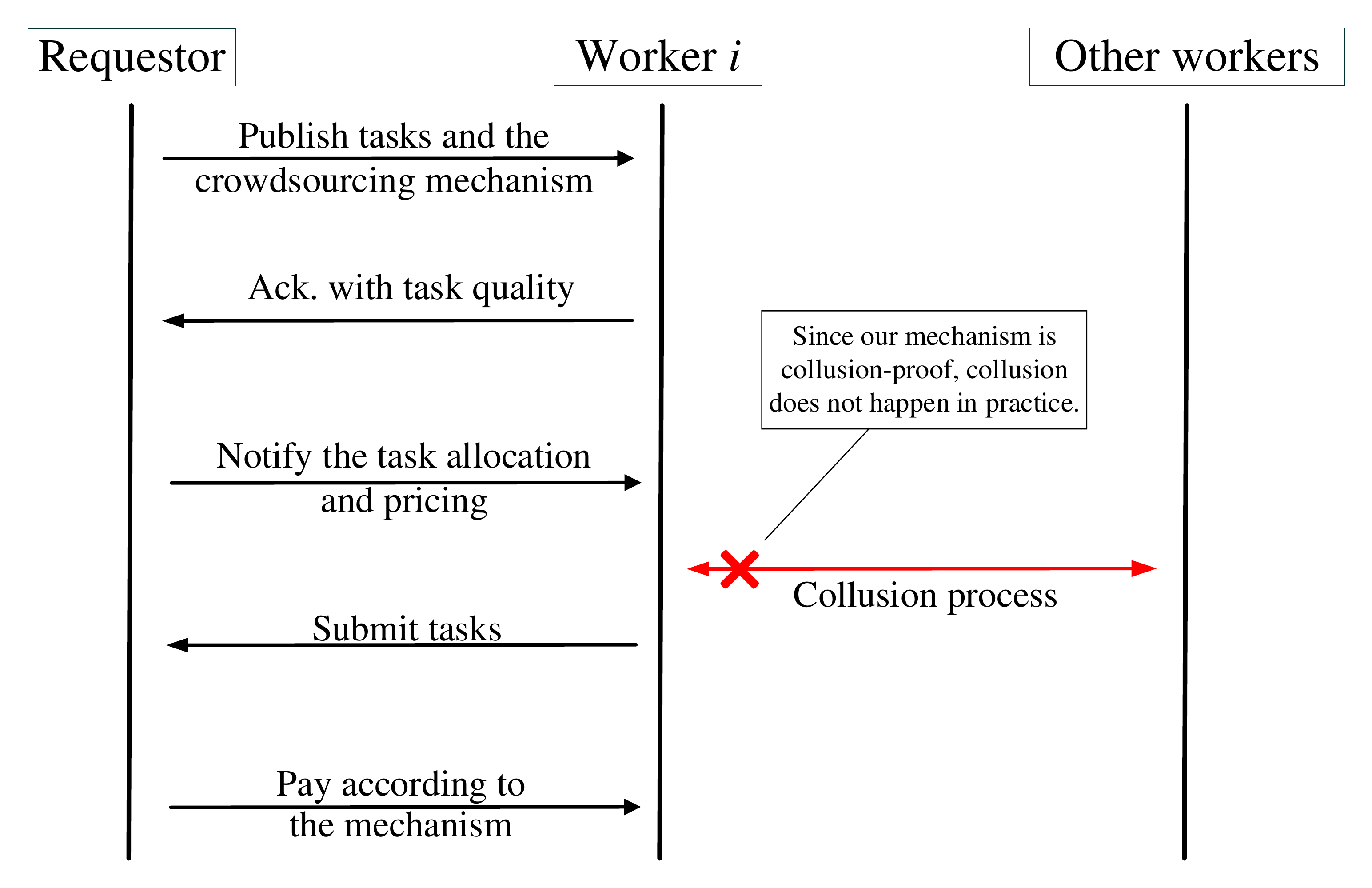}
	\caption{Overview of our framework.}
	\label{framework}
\end{figure}
The overview of our  framework is shown in Fig.\ref{framework}, which includes the following steps:
\begin{enumerate}
	\item The requestor publishes the tasks and the mechanism $\mathbb{G}$ online.
	\item If any candidate worker  $i$ accepts the tasks and the mechanism $\mathbb{G}$, he acknowledges  with his task quality to the requestor. Otherwise, he just ignores the messages from the requestor.
	\item Once receiving feedback from workers,  the requestor notifies them of the task allocation and pricing results.
	\item If the malicious workers decide to collude, they will negotiate the collusion process with each other. However, due to our collusion-proof mechanism, their optimal strategy is  not to collude.
	\item After finishing the tasks, the worker submits them to the requestor according to the requirements of $\mathbb{G}$.
	\item Accordingly, the requestor pays each worker in light of  the quantity and quality of the submitted tasks.
\end{enumerate}

We will detail our mechanism for different scenarios in the following two sections.
\section{Mechanism for the two-worker model}\label{sec:two}
In this paper, we conduct the analysis following a principle of simple to complex. We first consider the scenario where there is one requestor, denoted by $\mathbb{R}$, and two workers, denoted by $\mathbb{W}_i~ (i\in\{1,2\})$, in crowdsourcing.
In our model, affected by task difficulty, equipment characteristics, personal capability and environment\cite{peng2018data}, we consider any worker $\mathbb{W}_i$ completes tasks with quality $x_i \in \{\bar{x},\underline{x}\}$, with $\bar{x}$ and $\underline{x}$ respectively representing high and low task qualities.
Since it is difficult or costly for workers to change task qualities by adjusting their own capabilities or other factors such as the environment in the crowdsourcing process, the value of $x_i~ (i\in\{1,2\})$ is fixed. That is, a high-quality worker cannot complete low-quality tasks, and a low-quality worker is also unable to submit tasks with high quality.
We denote the probability that a worker submits high quality tasks with $p(\bar{x})=p$, and then $p(\underline{x})=1-p$. We assume the joint probability of task qualities submitted by the two workers as $p(x_1,x_2)$. 
For simplicity, let $p(\bar{x},\bar{x})=p_1,~ p(\bar{x},\underline{x})=p(\underline{x},\bar{x})=\frac{1}{2}p_2,~ p(\underline{x},\underline{x})=p_3$, where $p(\bar{x},\underline{x})=p(\underline{x},\bar{x})$. Apparently, $p_1=p^2$, $p_2=2p(1-p)$, $p_3={(1-p)}^2$.

To take advantage of the market power to deter misreport and collusion, the requestor regulates the behavior of workers through pricing and task allocation, which are determined according to different combinations of task qualities claimed by the two workers.   To make notations simple, we denote the cases where both workers claim $\bar{x}$, one claims $\bar{x}$ but the other claims $\underline{x}$, and both claim $\underline{x}$ as $X_1, X_2$, and $X_3$, respectively. In case $X_j~(j\in \{1,2,3\})$,
any $\mathbb{W}_i$ can obtain the payment  $u_{ij}=n_j\alpha(x_i)+t_j$  from the requestor at the working cost $F(n_j,x_i)$, where $n_j$ and  $t_j$ are respectively the number of tasks allocated to each worker and the extra reward used to incentivize workers to participate in crowdsourcing honestly;   $\alpha(\cdot)$ is the unit payoff for each task, which is a function of task quality.  In addition, $F(\cdot)$ is the cost function related to the task quality $x_i$ and the number of tasks $n_j$ which is the common knowledge.

Based on the prior beliefs, the requestor needs to propose a crowdsourcing mechanism $\mathbb{G}$ to the workers, which maps any set of quality distribution $\{p,\bar{x},\underline{x}\}$ into a set $\{n_j,t_j\}~(j\in\{1,2,3\})$. The aim of $\mathbb{G}$ is to maximize the expected utility of the requestor. That is,
\begin{equation}\label{eq:5}
	\max \sum_{j=1}^3p_j(R_j-\sum_{i=1}^2u_{ij}).
\end{equation}
In \eqref{eq:5}, $R_j$ is the payoff  of the requestor   obtained    from the submissions of two workers in case $X_j~(j\in\{1,2,3\})$.

Additionally, to realize the misreport-  and collusion-proof crowdsourcing, $\mathbb{G}$ should include the following constraints.

\subsection{Participation-incentive and \textbf{IC} constraints}\label{No-Colluison Mechanism}
In order to improve the completion rate and quality of tasks, the mechanism $\mathbb{G}$ should make workers willing to participate in the crowdsourcing and report the task quality truthfully. That is, the participation-incentive and \textbf{IC} constraints should be satisfied in the mechanism $\mathbb{G}$.  In the following, we first formulate the participation-incentive constraint and then \textbf{IC} one, where $\alpha(x_i)=x_i$ for simplicity\footnote{Our methodology can be applied to other forms of payoff function $\alpha(\cdot)$. The aim of simplifying $\alpha(\cdot)$ here is to make readers understand our design easily.}.

The main idea of the participation-incentive constraint is to guarantee  that a worker has no loss  when he participates in crowdsourcing. Specifically, if a worker is with  high quality, the other worker can be either high-quality or low-quality, corresponding to two cases $X_1$ and $X_2$. Hence, the expected utility of a high-quality worker consists of two parts: utility in $X_1$ with  probability $p_1$ and that in $X_2$ with probability $\frac{1}{2}p_2$. In order to motivate workers to participate in crowdsourcing, it is necessary to guarantee their expected utility no less than 0, that is,
\begin{equation}\label{eq:3}
	p_1(n_1\bar{x}+t_1-F(n_1,\bar{x}))+\frac{1}{2}p_2(n_2\bar{x}+t_2-F(n_2,\bar{x}))\geq 0.
\end{equation}

Similarly, for a low-quality worker, the participation-incentive constraint is
\begin{equation}\label{eq:4}
	\frac{1}{2}p_2(n_2\underline{x}+t_2-F(n_2,\underline{x}))+p_3(n_3\underline{x}+t_3-F(n_3,\underline{x}))\geq 0. \end{equation}

Next, we introduce the \textbf{IC} constraint. As mentioned above, the flexibility in the task number and extra reward determined by the requestor can  be used to deter any worker to lie. If a worker doesn't report his task quality truthfully, the mechanism is designed to make him cannot get the utility as  expected. That is to say, this mechanism makes the utility of an honest worker not less than that of a worker who misrepresents task quality. We take a high-quality worker as an example. When he reports truthfully, his expected utility is the linear combination of the utility in case $X_1$ and that in case $X_2$. When he lies about task quality, the utility in  case $X_1$ turns into that in case $X_2$ while the utility in  case $X_2$ becomes that in $X_3$. Hence, the \textbf{IC} constraint of a high-quality worker can be written as
\begin{equation}\label{eq:1}
	\begin{split}
		&p_1(n_1\bar{x}+t_1-F(n_1,\bar{x}))+\frac{1}{2}p_2(n_2\bar{x}+t_2-F(n_2,\bar{x}))\geq \\
		&p_1(n_2\underline{x}+t_2-F(n_2,\bar{x}))+\frac{1}{2}p_2(n_3\underline{x}+t_3-F(n_3,\bar{x})).
	\end{split}
\end{equation}
Similarly, for a low-quality worker, the \textbf{IC} constraint is
\begin{equation}\label{eq:2}
	\begin{split}
		&\frac{1}{2}p_2(n_2\underline{x}+t_2-F(n_2,\underline{x}))+p_3(n_3\underline{x}+t_3-F(n_3,\underline{x}))\geq \\
		&\frac{1}{2}p_2(n_1\bar{x}+t_1-F(n_1,\underline{x}))+p_3(n_2\bar{x}+t_2-F(n_2,\underline{x})).
	\end{split}
\end{equation}

It is worth noting that the \textbf{IC} constraint for the high-quality worker presented in \eqref{eq:1} is indispensable even it seems very unlikely that a high-quality worker will pretend to be low-quality.
In fact, as a rational and utility-driven player, once the high-quality worker finds it more profitable to be low-quality, he will definitely lie about his quality level to pursue utility maximization. This phenomenon is highly possible to happen when there is no guarantee of \textbf{IC} constraint for the high-quality worker. Because in this case, the requestor will put more consideration on how to avoid the problem of low-quality worker's lying, that is, how to make the honest low-quality workers gain more payoff than the dishonest ones. And when the probability of high-quality submission $p(\bar{x})$ is small, in order to prevent the low-quality worker from lying, the mechanism designed by the requestor will inevitably enforce greater number of tasks and larger amount of extra reward on the low-quality worker, which will in turn result in an unexpected situation where the high-quality worker will get more profit when lying. For example,  when $\bar{x}=9,\underline{x}=1$ and $n_1=1,n_2=4,n_3=3$, the optimal value of $t_1, t_2$ and $t_3$ can be expressed as $1.235$, $-3.32p+99.35$ and $-210.579p+143.76$ through linear regression\footnote{Because any arbitrary set of $\{n_1,n_2,n_3\}$  responds to an optimal set of $\{t_1, t_2,t_3\}$ maximizing the utility of the requestor with corresponding constraints, we directly set the values of $\{n_1,n_2,n_3\}$  here for  simplicity. In addition, the optimal set of $\{t_1, t_2,t_3\}$  here is solved with the participation-incentive constraint, \textbf{IC} constraint for the the low-quality worker and the collusion-proof constraint which we will discuss later, but without the \textbf{IC} constraint for the high-quality worker to identify the necessity of it.}. Substituting these expressions into \eqref{eq:1}, we can find that when $p<0.4$,  the high-quality worker will choose to lie, which is obviously unexpected in practice. Thus, to eliminate this problem, we have to include the \textbf{IC} constraint for the high-quality worker as shown in \eqref{eq:1}.

\subsection{Collusion-proof constraint}
Generally, a common purpose of the two workers is to maximize their total expected utility, i.e.,
\begin{equation}\label{eq:6}
	\max \sum_{j=1}^3\sum_{i=1}^2p_j(u_{ij}-F(n_j,x_i)).
\end{equation}

When the two workers are both high-quality, they obviously have no need to behave maliciously. While if both are  low-quality, the \textbf{IC} constraint makes no room for them to pretend to be high-quality so that the requestor has to pay based on the low quality no matter who completes the tasks. In other words, there is no chance to collude when they are both low-quality workers.  Hence, the collusion will only occur when the task  qualities of the two workers are different, where the low-quality worker completes  a part of or even the whole task which is supposed to be finished by the high-quality worker, so that the low-quality worker can earn more reward while the high-quality one is able to reduce the working cost. Such a malicious behavior cannot be prevented by the \textbf{IC} constraint, pressing a need to design the collusion-proof constraint for the mechanism $\mathbb{G}$, which is detailed as follows.

Without collusion, the total utility of the workers is $n_2(\bar{x}+\underline{x})+2t_2-F(n_2,\bar{x})-F(n_2,\underline{x})$. If  the high-quality worker colludes with the low-quality one through assigning $k$ units of task  to the low-quality worker to complete, their total utility would become $n_2(\bar{x}+\underline{x})+2t_2-F(n_2-k,\bar{x})-F(n_2+k,\underline{x})~(k\in \{0,1,\cdots,n_2\}$). In order to prevent collusion, our mechanism $\mathbb{G}$ should be designed to make sure that their total utility without collusion will never be less than that with collusion. In other words, the following inequation should be satisfied
\begin{align}\label{eq:7}
	F(n_2,\bar{x})+F(n_2,\underline{x}) & \leq
	F(n_2-k,\bar{x})+F(n_2+k,\underline{x}),\\ & k\in \{0,1,\cdots,n_2\}. \notag
\end{align}

For ease of calculation and analysis, let $F(n,x)=f(n)x$. As $f(n)x$ is the cost required for a worker to complete tasks, the function $f(\cdot)$ should be an non-negative and non-decreasing function in $[n_2-k,n_2+k]$.

\begin{proposition}
	The  mechanism  $\mathbb{G}$ entails that if $f(\cdot)$ is a monotonically increasing function, i.e., $f(\cdot)'>0$, in $[n_2-k,n_2+k]$, it is  also a concave one, i.e., $f(\cdot)''>0$.
\end{proposition}

\begin{proof}
	According to \eqref{eq:7},
	\begin{equation}\label{eq:8}
		\begin{split}
			f(n_2)\bar{x}+f(n_2)\underline{x}\leq f(n_2-k)\bar{x}+f(n_2+k)\underline{x}\\
			\Rightarrow (f(n_2)-f(n_2-k))\bar{x}\leq (f(n_2+k)-f(n_2))\underline{x} \\
			k\in \{0,1,\cdots,n_2\}.
		\end{split}
	\end{equation}
	If $f(\cdot)'>0$, that is, $f(\cdot)$ is monotonically increasing, then
	\begin{equation}\label{eq:9}
		\frac{f(n_2)-f(n_2-k)}{f(n_2+k)-f(n_2)}\leq\underline{x}/\bar{x},~k\in \{0,1,\cdots,n_2\}.
	\end{equation}
	Obviously, $f(\cdot)$ is a concave function, i.e., $f(\cdot)''>0$.
\end{proof}
Under the assumption of $f(\cdot)'>0$, we can see that the larger the value of $k$ (the number of collusive tasks), the smaller the workers' cost. Therefore the workers' optimal collusion strategy is that all the tasks of the  high-quality worker are completed by the low-quality one, i.e., $k=n_2$. In order to avoid collusion, our mechanism should satisfy the following constraint
\begin{equation}\label{eq:10}
	\frac{f(n_2)-f(0)}{f(2n_2)-f(n_2)}\leq \underline{x}/\bar{x}.
\end{equation}
\eqref{eq:10} guarantees that even though the malicious workers adopt the optimal collusion strategy, the total expected utility of colluded workers cannot be larger than that when they behave honestly. Thus, collusion-proof is realized in crowdsourcing.

It should be noted that the assumption of $f(\cdot)'>0$  is consistent with most practical scenarios, where  as the number of tasks increases, the cost of workers grows. 
\section{Mechanism for the $m$-worker model}\label{sec:m}
In this section, we discuss the $m$-worker ($m > 2$) scenario which is a general case and can be derived from the above optimization process.

In a general model, there is one requestor $\mathbb{R}$ and $m$ workers $\mathbb{W}_1,\mathbb{W}_2,\cdots,\mathbb{W}_m$, whose task qualities are denoted as $x_1,x_2, \cdots, x_m$. We use case $X_j$ $(j\in \{1, 2, \cdots, m, m+1\})$ to represent the situation where there are $m+1-j$ high-quality workers among the $m$  workers. $p_j$ is the probability of case $X_j$ and it is easy to know that $p_j=\mathbb{C}_m^{m+1-j}p^{m+1-j}(1-p)^{j-1}$ due to $p(\bar{x})=p$, where $\mathbb{C}_m^i$ is the combination calculation $\mathbb{C}_m^i=\frac{i!}{m!(m-i)!}$. The mechanism $\mathbb{G}$ proposed by the requestor maps any set of quality distribution $\{p,\bar{x},\underline{x}\}$ into a binary group $\{n_j,t_j\}$, where $n_j$ and $t_j$ are the task number and extra reward in case $X_j$ $(j\in \{1, 2, \cdots, m, m+1\})$. In the $j^{th}$ case, $\mathbb{W}_i$ spends $F(n_j,x_i)$ and earns $u_{ij}=n_jx_i+t_j$ from the requestor.
Similar to Section \ref {sec:two}, the aim of $\mathbb{G}$ is to maximize the requestor's expected utility. That is,
\begin{equation}\label{eq:m}
	\max \sum_{j=1}^{m+1}p_j(R_j-\sum_{i=1}^{m}u_{ij}).
\end{equation}
In \eqref{eq:m}, $R_j$ is the payoff of the requestor obtained from the submissions of all workers in case $X_j$ and $u_{ij}$ is the payment to worker $i$ in case $X_j$ $(j\in \{1, 2, \cdots, m, m+1\})$.

To realize the misreport-  and collusion-proof crowdsourcing, $\mathbb{G}$ should meet the \textbf{IC}, the participation-incentive and the collusion-proof constraints, which are described in the following.

For any high-quality worker, his expected utility is related to the number of other high-quality workers in crowdsourcing. In detail, when there are total $m-i$ high-quality workers including himself, case $X_i$ occurs with the probability of $\mathbb{C}_{m-1}^{m-i}p^{m+1-i}(1-p)^{i-1}=\frac{\mathbb{C}_{m-1}^{m-i}}{\mathbb{C}_{m}^{m+1-i}}p_i$. Thus, in case $X_i$ , the expected utility of this high-quality worker is $\frac{\mathbb{C}_{m-1}^{m+1-i}}{\mathbb{C}_{m}^{m+1-i}}p_i(n_i\underline{x}+t_i-F(n_i,\underline{x}))$. However, when the high-quality worker feigns a low-quality one, the requestor will consider there are $m-(i+1)$ high-quality workers, and hence, the number of tasks and extra reward to this worker are respectively $n_{i+1}$ and $t_{i+1}$. Thus,  the expected utility of this high-quality worker turns to be $\frac{\mathbb{C}_{m-1}^{m-i}}{\mathbb{C}_{m}^{m+1-i}}p_i(n_{i+1}\underline{x}+t_{i+1}-F(n_{i+1},\bar{x}))$.
Considering $i\in \{1,2,\dots,m\}$, the \textbf{IC}  and the participation-incentive constraints can be formulated by \eqref{eq:19} and \eqref{eq:20} as follows, implying that the expected utility of a high-quality worker when he honestly participates   in crowdsourcing is not less than that when he misreports the task quality to the requestor or even does not join in the crowdsourcing.
\begin{equation}\label{eq:19}
	\begin{split}
		\sum_{i=1}^m\frac{\mathbb{C}_{m-1}^{m-i}}{\mathbb{C}_{m}^{m+1-i}}p_i(n_i\bar{x}+t_i-F(n_i,\bar{x}))\geq \\ \sum_{i=1}^m\frac{\mathbb{C}_{m-1}^{m-i}}{\mathbb{C}_{m}^{m+1-i}}p_i(n_{i+1}\underline{x}+t_{i+1}-F(n_{i+1},\bar{x})),
	\end{split}
\end{equation}
\begin{equation}\label{eq:20}
	\sum_{i=1}^m\frac{\mathbb{C}_{m-1}^{m-i}}{\mathbb{C}_{m}^{m+1-i}}p_i(n_i\bar{x}+t_i-F(n_i,\bar{x}))\geq 0.
\end{equation}

Similarly, for a low-quality worker, his \textbf{IC} and participation-incentive constraints are respectively written as
\begin{equation}\label{eq:21}
	\begin{split}
		\sum_{i=2}^{m+1}\frac{\mathbb{C}_{m-1}^{m+1-i}}{\mathbb{C}_{m}^{m+1-i}}p_i(n_i\underline{x}+t_i-F(n_i,\underline{x}))\geq \\ \sum_{i=2}^{m+1}\frac{\mathbb{C}_{m-1}^{m+1-i}}{\mathbb{C}_{m}^{m+1-i}}p_i(n_{i-1}\bar{x}+t_{i-1}-F(n_{i-1},\bar{x})),
	\end{split}
\end{equation}
\begin{equation}\label{eq:22}
	\sum_{i=2}^{m+1}\frac{\mathbb{C}_{m-1}^{m+1-i}}{\mathbb{C}_{m}^{m+1-i}}p_i(n_i\underline{x}+t_i-F(n_i,\underline{x}))\geq 0.
\end{equation}

Next, we will analyze the issue of collusion. As in the case of two workers, collusion does not occur when the qualities of all workers are the same, so we only consider the collusion in the case where the workers' task qualities are different. Take three workers as an example, we can find that collusion may occur in two cases, namely the case that there are 2 high-quality workers represented by $\{\bar{x},\bar{x},\underline{x}\}$, as well as the case that there is 1 high-quality worker represented by $\{\bar{x},\underline{x},\underline{x}\}$\footnote{$\{\bar{x},\bar{x},\underline{x}\}$ does not indicate the task quality of any specific worker is low or high but represents all the scenarios where there are 2 high-quality workers and 1 low-quality one. So does $\{\bar{x},\underline{x},\underline{x}\}$.}. In the case of $\{\bar{x},\bar{x},\underline{x}\}$, the total utility of all workers without collusion is $n_2(2\bar{x}+\underline{x})+3t_2-2F(n_2,\bar{x})-F(n_2,\underline{x})$. Assume that $\mathbb{W}_1$ and $\mathbb{W}_2$ are both the high-quality workers while $\mathbb{W}_3$ is the low-quality one, we can find that the costs of  $\mathbb{W}_1$ and $\mathbb{W}_2$ are both $F(n_2,\bar{x})$ while $\mathbb{W}_3$'s cost is $F(n_2,\underline{x})$.
When $\mathbb{W}_1$ and $\mathbb{W}_2$ respectively require $\mathbb{W}_3$ to complete $k_1 \in[0,n_2]$ and $k_2\in[0,n_2]$ tasks, their costs change to $F(n_2-k_1,\bar{x})+$ collusion expense\footnote{The collusion expense is the reward the high-quality worker paid to the low-quality one in order to facilitate collusion which including the low-quality worker's cost of completing the collusion tasks and the collusion bribe.}, $F(n_2-k_2,\bar{x})+$ collusion expense, and $F(n_2,\underline{x})+F(n_2,\underline{x})+F(n_2,\underline{x})$ respectively. The workers' total utility after collusion is $n_2(2\bar{x}+\underline{x})+3t_2-F(n_2-k_1,\bar{x})-F(n_2-k_2,\bar{x})-F(n_2+k_1+k_2,\underline{x})$. Through comparing the total utility before and after collusion, the collusion-proof constraint can be written as
\begin{equation}\label{eq:15}
	\begin{split}
		2&F(n_2,\bar{x})+F(n_2,\underline{x})\leq \\ F(n_2-k_1,\bar{x})+&F(n_2-k_2,\bar{x})+F(n_2+k_1+k_2,\underline{x}), \\
		&k_1,k_2 \in \{0,1,\cdots,n_2\}.
	\end{split}
\end{equation}

Likewise, in the case of $\{\bar{x},\underline{x},\underline{x}\}$, if there is no collusion, the cost of $\mathbb{W}_1$ is $F(n_3,\bar{x})$, while those of $\mathbb{W}_2$ and $\mathbb{W}_3$ are both $F(n_3,\underline{x})$. When $\mathbb{W}_1$ requires $\mathbb{W}_2$ and $\mathbb{W}_3$ to complete $k_1$ and $k_2$ tasks respectively ($k_1+k_2 \in \{0,1,\cdots,n_3\}$), their costs change to $F(n_3-k_1-k_2,\bar{x})+$ collusion expense, $F(n_3,\underline{x})+F(k_1,\underline{x})$ and $F(n_3,\underline{x})+F(k_2,\underline{x})$. The collusion-proof constraint can be written as
\begin{equation}\label{eq:16}
	\begin{split}
		&F(n_3,\bar{x})+2F(n_3,\underline{x})\leq \\ F(n_3-k_1-k_2,&\bar{x})+F(n_3+k_1,\underline{x})+F(n_3+k_2,\underline{x}), \\
		&k_1+k_2\in \{0,1,\cdots,n_3\}.
	\end{split}
\end{equation}
\eqref{eq:15} and \eqref{eq:16} guarantee that  the total utility of all workers when they collude will not be greater than that when they behave honestly in the above two cases, so that the workers have no motivation to collude.

With the same assumption $F(n,x) = f(n)x$ mentioned in Section \ref{sec:two}, when $f(\cdot)$ is monotonically increasing, $f(\cdot)''\geq 0$ can be deduced with the similar method used above. In the  case of $\{\bar{x},\bar{x},\underline{x}\}$, when collusion occurs, it is obvious that the total cost of the workers (i.e., the right part of the inequality) minimizes at $k_1=k_2=n_2$, because the cost saved by collusion is the most. And in the case of $\{\bar{x},\underline{x},\underline{x}\}$, it is also obvious that the total cost minimizes at $k_1+k_2=n_3$. According to the nature of the concave function $f(x_1)+f(x_2)\geq 2f(\frac{x_1+x_2}{2})$, we have $f(n_3+k_1)+f(n_3+k_2)\geq 2f(\frac{n_3+k_1+n_3+k_2}{2})=2f(\frac{3n_3}{2})$. So in this case, the total cost minimizes at $k_1=k_2=\frac{1}{2}n_3$. Therefore, \eqref{eq:15} and \eqref{eq:16} can be transferred to
\begin{equation}\label{eq:17}
	2f(n_2)\bar{x}+f(n_2)\underline{x}\leq 2f(0)\bar{x}+f(3n_2)\underline{x},
\end{equation}
\begin{equation}\label{eq:18}
	f(n_3)\bar{x}+2f(n_3)\underline{x}\leq f(0)\bar{x}+2f(\frac{3}{2}n_3)\underline{x},
\end{equation}
which implies that the maximum utility of any worker that can be achieved by collusion is less than that without collusion. At this point, workers will give up collusion.

For $m$ workers, there are $m-1$ cases in total will incur collusion, which could be represented by $\{\bar{x},\bar{x},\cdots,\bar{x},\underline{x}\},\{\bar{x},\bar{x},\cdots,\bar{x},\underline{x},\underline{x}\},\cdots,\{\bar{x},\underline{x},\cdots,\underline{x}\}$. And in the $i^{th}$ case, the number of high-quality workers is $(m-i)$, then the total utility of all workers in case $i$ before collusion is
\begin{equation}\label{eq:23}
	\begin{split}
		\mathbb{U}_{i}=&n_{i+1}((m-i)\bar{x}+i\underline{x})+mt_{i+1}\\
		-&(m-i)F(n_{i+1},\bar{x})-i \cdot F(n_{i+1},\underline{x}).\\
	\end{split}
\end{equation}
Based on the analysis above, we still let $F(n,x) = f(n)x$ and assume that $f(\cdot)$ is a monotonically increasing and concave function. When the workers collude, their optimal strategy is that all tasks of high-quality workers are equally assigned to low-quality workers due to the nature of concave function $\sum_i^n f(x_i)\geq nf(\frac{\sum_i^n x_i}{n})$. Thus,  the total  utility of all the malicious workers after using the optimal collusion strategy  is
\begin{equation}\label{eq:24}
	\begin{split}
		\mathbb{U}'_{i}=&n_{i+1}((m-i)\bar{x}+i\underline{x})+mt_{i+1}\\
		-&(m-i)f(0)\bar{x}-i \cdot f(n_{i+1}+\frac{m-i}{i}n_{i+1})\underline{x}.\\
	\end{split}
\end{equation}
Based on \eqref{eq:23} and \eqref{eq:24}, the collusion-proof constraint for the $i^{th}$ case is
\begin{equation}\label{eq:25}
	\begin{split}
		(m-i)f(&n_{i+1})\bar{x}+i \cdot f(n_{i+1})\underline{x}\leq \\ (m-i)f(0)&\bar{x}+i \cdot f(n_{i+1}+\frac{m-i}{i}n_{i+1})\underline{x},\\
		&i \in \{1,2,\cdots,m-1\}
	\end{split}
\end{equation}
which implies even though the malicious workers adopt the optimal collusion strategy, their total utility is still not larger than that when they do not collude. Thus, the goal of collusion-proof can be achieved in the $m$-worker model.
\section{Performance evaluation}\label{sec:performance}
In this section, we analyze the impacts of some key parameters on the performance of our proposed crowdsourcing mechanism through extensive simulations.
\subsection{Two workers}
We first analyze the situation of two workers presented in Section \ref{sec:two}. Basically, we utilize the following set of functions\footnote{We also test other function settings which present similar results. Hence, we do not present them for brevity.} to calculate \eqref{eq:5} to \eqref{eq:2}, and \eqref{eq:10}: $f(n)=2^n-1$, $R_1=2 \bar{x}\times n_1^2$, $R_2=(\bar{x}+\underline{x})\times n_2^2$, $R_3=2\underline{x}\times n_3^2$.
Considering that the actual meaning of $n_i~(i\in\{1,2,3\})$ is the number of tasks assigned to each worker under three different cases $X_i~(i\in\{1,2,3\})$, we restrict them as positive integers.

To begin with, we investigate the change laws of the requestor's optimal utility in Fig. \ref{fig:optimal_utility}. As shown in Fig. \ref{y4}, we present the optimal utility changing with different $p(\bar{x})$ and $\bar{x}$ given a fixed $\underline{x}$; in Fig. \ref{y5}, we report the change of the optimal utility with various $p(\bar{x})$ and $\underline{x}$ when $\bar{x}$ is fixed.
From these two figures, one can see that when $\bar{x}$ or $\underline{x}$ is settled, the requestor's optimal utility increases with the increasing $p(\bar{x})$. The reason is that the greater  $p(\bar{x})$ is, the more high-quality tasks the requestor receives, which finally increases the requestor's utility.
Besides, when one of $\bar{x}$ and $\underline{x}$ is fixed, the larger the difference between high and low quality $\Delta x$,  the higher increment of the requestor's optimal utility due to the change of $p(\bar{x})$.
That is, when $\Delta x$ is increasing, $p(\bar{x})$ has more influence on the requestor's utility. The underlying reason is that the larger $\Delta x$ incurs the larger difference of utility brought by the high-quality worker and low-quality one, which will amplify the improvement of the optimal utility resulted from the increasing $p(\bar{x})$.
Last but not least, for a specific $p(\bar{x})$, one can observe that the requestor's optimal utility becomes smaller with the increasing $\bar{x}$, shown in Fig. \ref{y4}, but comes to be larger with the increasing $\underline{x}$, presented in Fig. \ref{y5}.
This is because when $\underline{x}$ is constant, the increase of $\bar{x}$ implies that collusion and misreport  turn to be more profitable, incentivizing more malicious behaviors from workers, which drives the requestor to spend more to eliminate these undesirable behaviors, leading to less utility for her. While when $\bar{x}$ is given, the larger $\underline{x}$ makes the profit of the low-quality worker's lying become trivial, dispelling his enthusiasm of behaving maliciously, which enables the increase of the requestor's optimal utility.

\begin{figure}[]
	\centering
	\subfigure[$\underline{x}=1$.]{
		\includegraphics[width=0.4\textwidth]{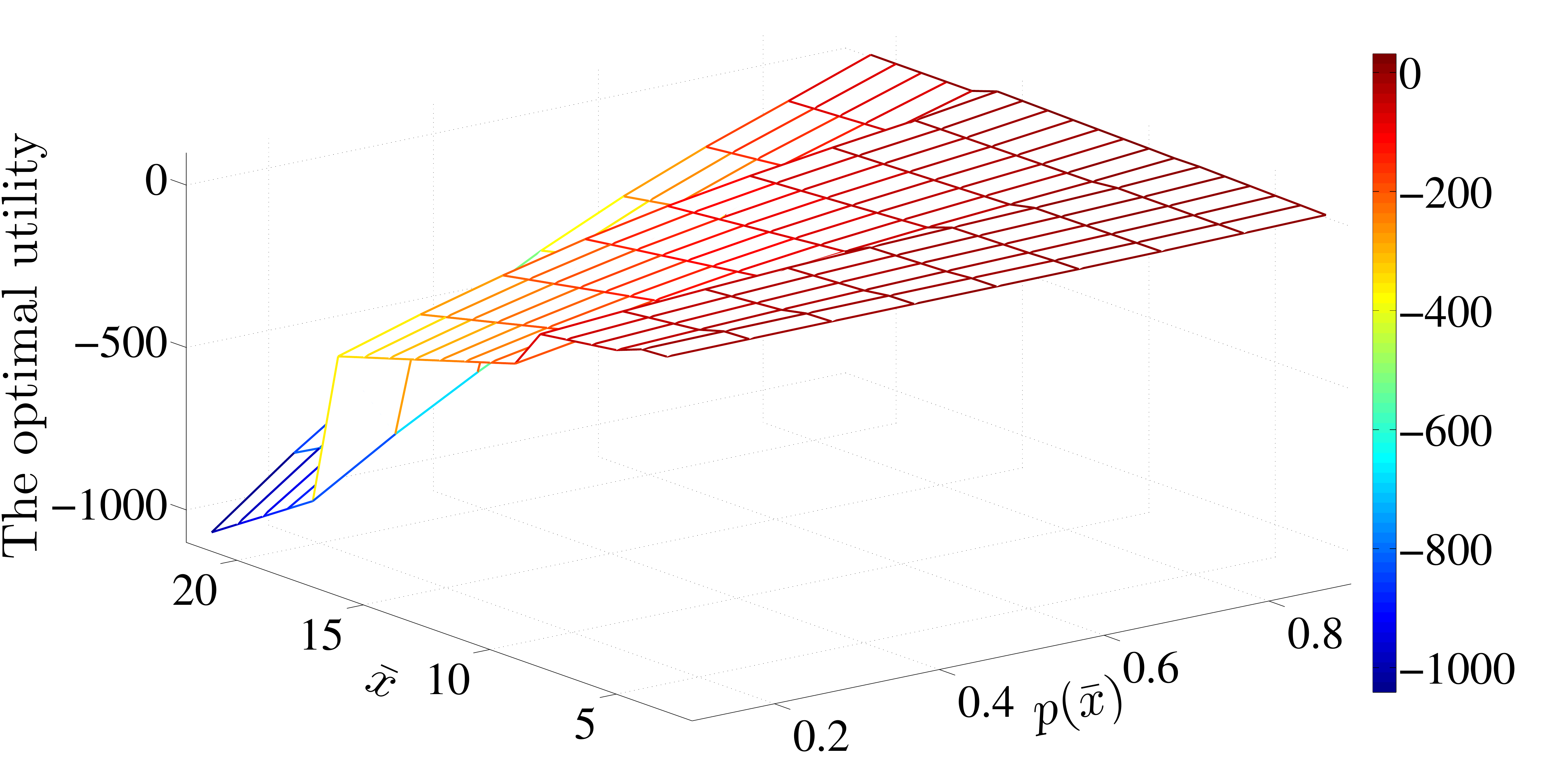}
		\label{y4}
	}
	\subfigure[$\bar{x}=21$.]{
		\includegraphics[width=0.4\textwidth]{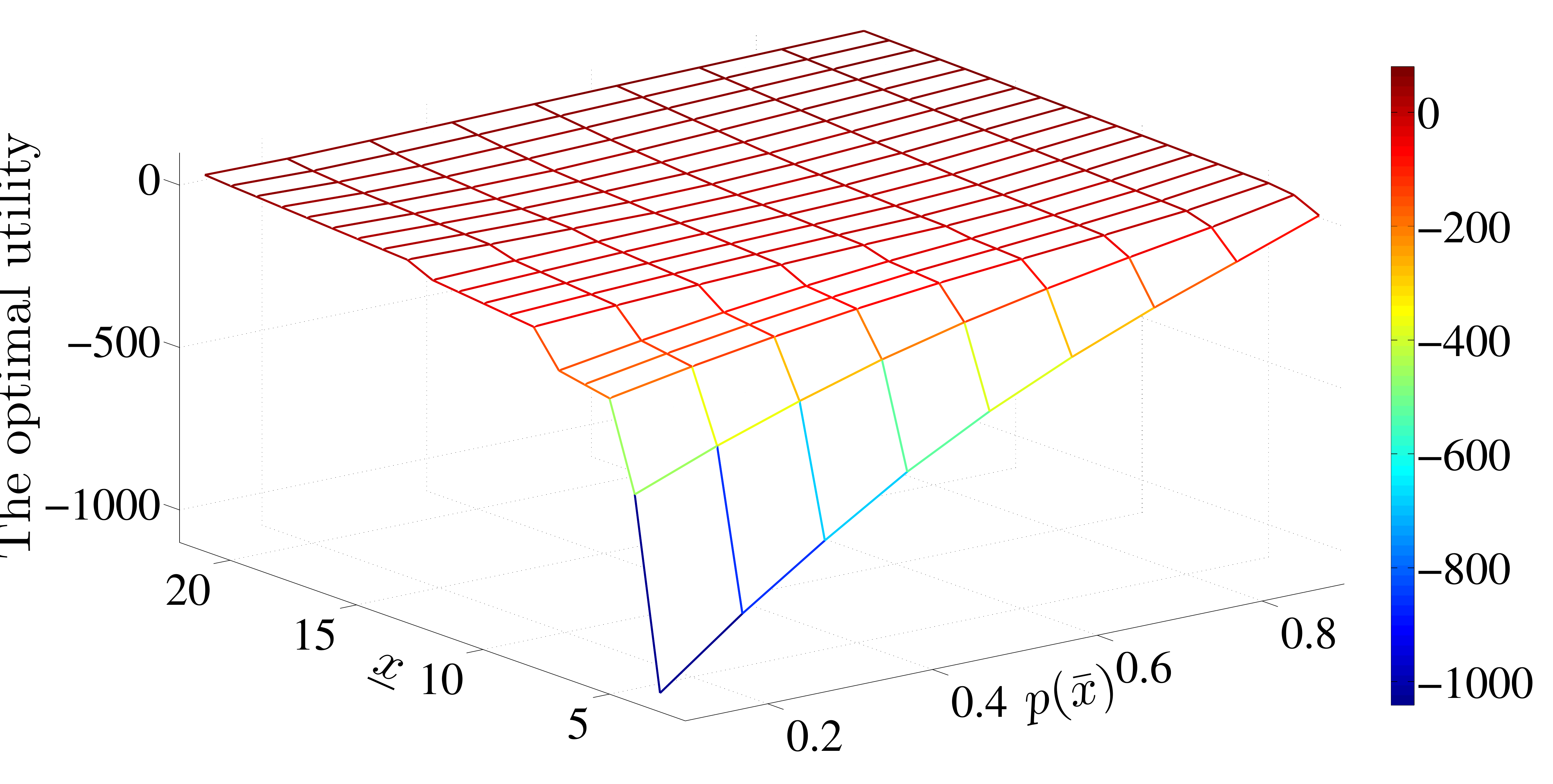}
		\label{y5}
	}
	\caption{The optimal utility of the requestor in different situations.}
	\label{fig:optimal_utility}
\end{figure}

Next, to explore the relationships between $n_i~(i \in \{1,2,3\})$ and  $\Delta x$, $p(\bar{x})$, $\bar{x}$, we solve the optimization problem with fixed $\Delta x$, $p(\bar{x})$ and $\bar{x}$, respectively, and analyze the corresponding experimental results. In TABLE \ref{nfigure}, we report the changing trend of $n_1$, $n_2$ and $n_3$ in different cases. The parameter settings in each case are clarified in this table, and the meanings of horizontal and vertical coordinates are specified in the corresponding figures.

From the first column of figures in TABLE \ref{nfigure}, we can see that the changing trend of $n_1$ has the following characteristics:
\begin{enumerate}[labelsep = .5em, leftmargin = 0pt, itemindent = 2.5em]
	\item As shown in the first two and the last two figures, the value of $n_1$ keeps stable or increasing with the increase of $p(\bar{x})$.
	This is because $n_1$ corresponds to the situation where both workers are of high quality, so the more likely the high-quality worker appears, the higher value of $n_1$ that the requestor needs to set to achieve the \textbf{IC} constraint for the high-quality worker and avoid any potential lie. 
	\item In the middle two figures, as $\bar{x}$ increases, $n_1$ remains stable or increases. 
	The reason is that the increase of $\bar{x}$ drives the requestor to increase $n_1$ for achieving the \textbf{IC} constraint. 
	\item In the first two figures, when $\underline{x}$ increases, one can find that $n_1$ remains unchanged or increases. To be specific, when $\Delta x$ is fixed, the increasing  $\underline{x}$ indicates the increase of $\bar{x}$, which makes the requestor enlarge $n_1$ to satisfy the \textbf{IC} constraint.
\end{enumerate}

Through comparing the second column of figures, we can find that $n_2$ changes with the following features:
\begin{enumerate}[labelsep = .5em, leftmargin = 0pt, itemindent = 2.5em]
	\item From the first two figures and the last two ones, it is obvious that $n_2$ keeps unchanged or becomes larger with the increasing $p(\bar{x})$. This is because $n_2$ represents the task allocation in the case of one high-quality worker and one low-quality worker, which affects the \textbf{IC} constraint and the collision-proof constraint at the same time. As $p(\bar{x})$ increases,  the likelihood that a worker is a  high-quality one increases, which makes the problem of high-quality worker's lying more serious. This situation is very similar to $n_1$, so the requestor needs to increase $n_2$ to guarantee the \textbf{IC} constraint.
	\item Observing the first two figures in this column, one can get the conclusion that $n_2$ owns a trend of decreasing first and then increasing with the increase of $\underline{x}$; and in all the rest of figures in this column, $n_2$ first decreases and then increases with the increasing $\Delta x$. This is because, on one hand, $n_2$ affects all constraints, making the requestor increase $n_2$ to achieve these constraints; on the other hand, the increase of $n_2$ will lead to more low-quality submissions for the requestor, lowering its utility, which is exactly opposite to the optimization direction. So $n_2$ finally shows a trend of decreasing first and then increasing. In detail, take the middle two figures as an example, with the fixed $p(\bar{x})$ and $\bar{x}$, when $\Delta x$ is small, $\underline{x}$ decreases as $\Delta x$ increases. In this case, due to the relatively small quality difference, the problems of workers' lying and collusion are not so serious, then the requestor can improve its utility through reducing the number of low-quality tasks with a lower $n_2$. When $\Delta x$ gradually increases, the problems of workers' lying and collusion become more severe, then $n_2$ comes to the minimum. When $\Delta x$ increases to the maximum, the workers' lying and collusion problem turns to be terrible, where the requestor needs to increase $n_2$ to satisfy all constraints. Therefore, $n_2$ presents a trend of decreasing first and then increasing.
\end{enumerate}

Finally, we scrutinize the third column of figures for $n_3$:
\begin{enumerate}[labelsep = .5em, leftmargin = 0pt, itemindent = 2.5em]
	\item As can be seen from all these figures, the overall change of $n_3$ is not obvious and remains generally unchanged. This is because, on the one hand, $n_3$ corresponds to the case of two low-quality workers, where increasing $n_3$ means getting more low-quality tasks, opposite to the optimization direction; on the other hand, reducing $n_3$ may result in failure to meet the \textbf{IC} constraint of the low-quality worker. These two conflicting factors are balanced in most cases, making $n_3$ relatively stable.
	\item As we can see from the first two and the last two figures, $n_3$ stays the same or decreases with the increase of $p(\bar{x})$. Because when $p(\bar{x})$ increases, the possibility that both workers are low-quality ones is reduced, so that the requestor can appropriately decrease $n_3$ to improve her utility by receiving less low-quality tasks.
	\item From the $3^{rd}$ figure with a lower $p(\bar{x})$, $n_3$ decreases first and then increases with the increasing $\Delta x$, which is because the probability of two low-quality workers is high, so the two factors mentioned in the first item cannot be balanced. To be specific, when $\Delta x$ is small, the benefit of low-quality worker's lying is low, where the cost for the requestor to meet the \textbf{IC} constraint is also small, so the requestor can lower $n_3$ to improve her utility. While when $\Delta x$ is large, the problem of the low-quality worker's lying becomes more rigorous than the low utility due to many low-quality submissions, so the requestor will increase $n_3$ to satisfy the \textbf{IC} constraint. These two aspects generally make $n_3$ decrease first and then increase.
\end{enumerate}

\begin{table*}[]
	\centering
	\caption{$n_1$, $n_2$ and $n_3$ in different situations.}
	\label{nfigure}
	\begin{tabular}{|c|c|c|c|}
		\hline
		& $n_1$ & $n_2$ & $n_3$ \\ \hline
		$\Delta x=4$ &  \begin{minipage}{0.28\textwidth}
			\includegraphics[width=5cm]{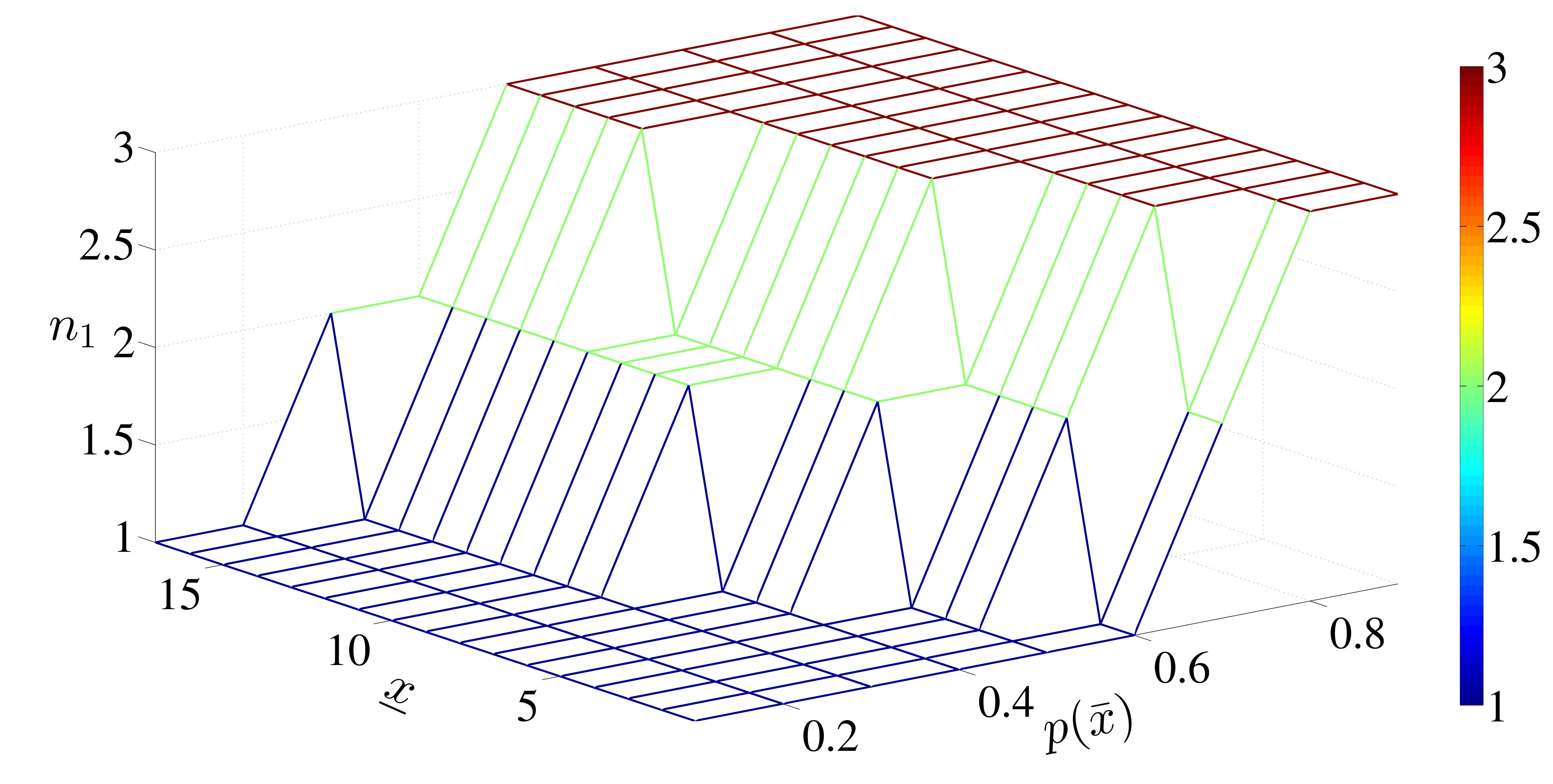}
		\end{minipage} & \begin{minipage}{0.28\textwidth}
			\includegraphics[width=5cm]{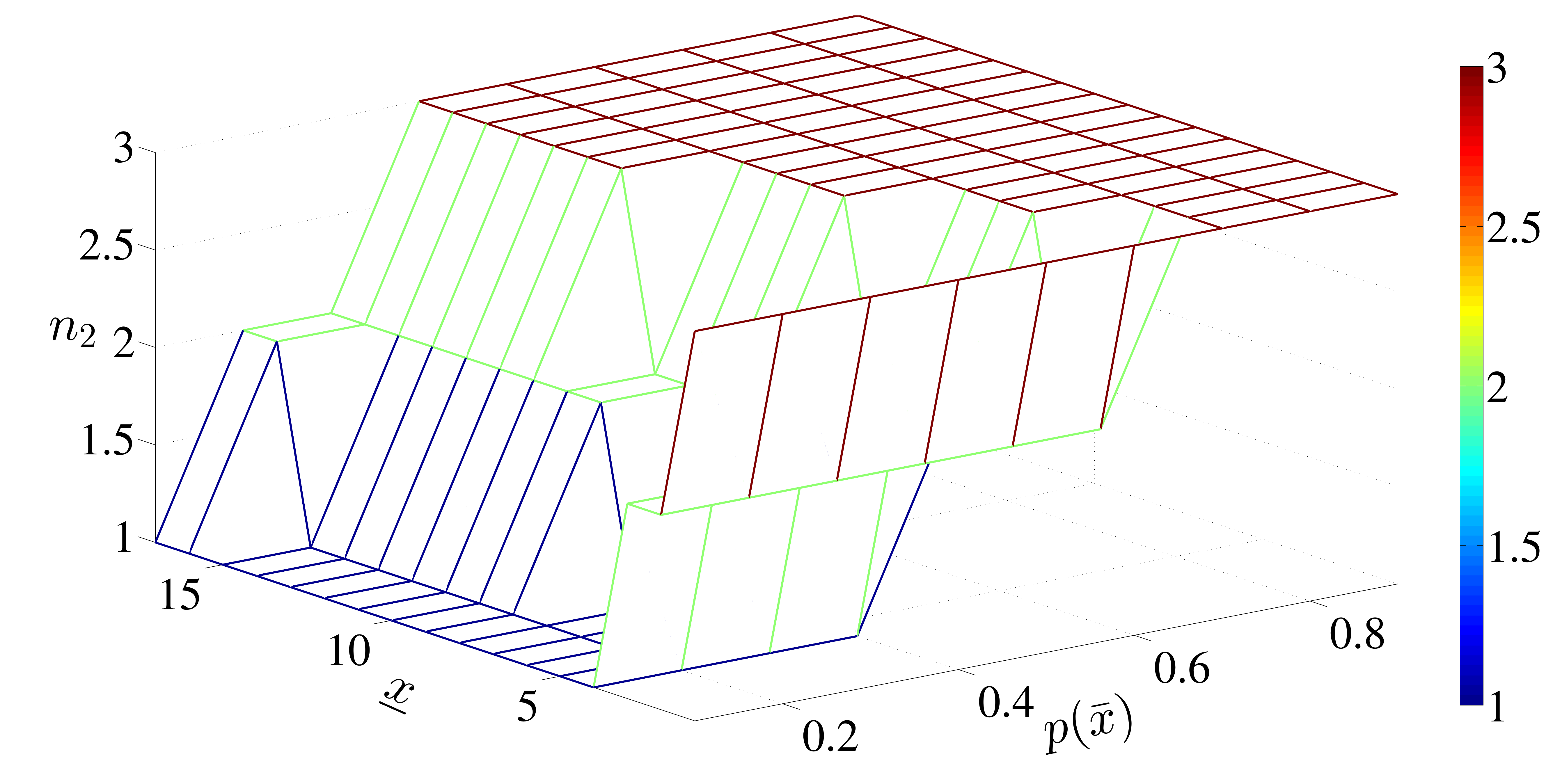}
		\end{minipage} & \begin{minipage}{0.28\textwidth}
			\includegraphics[width=5cm]{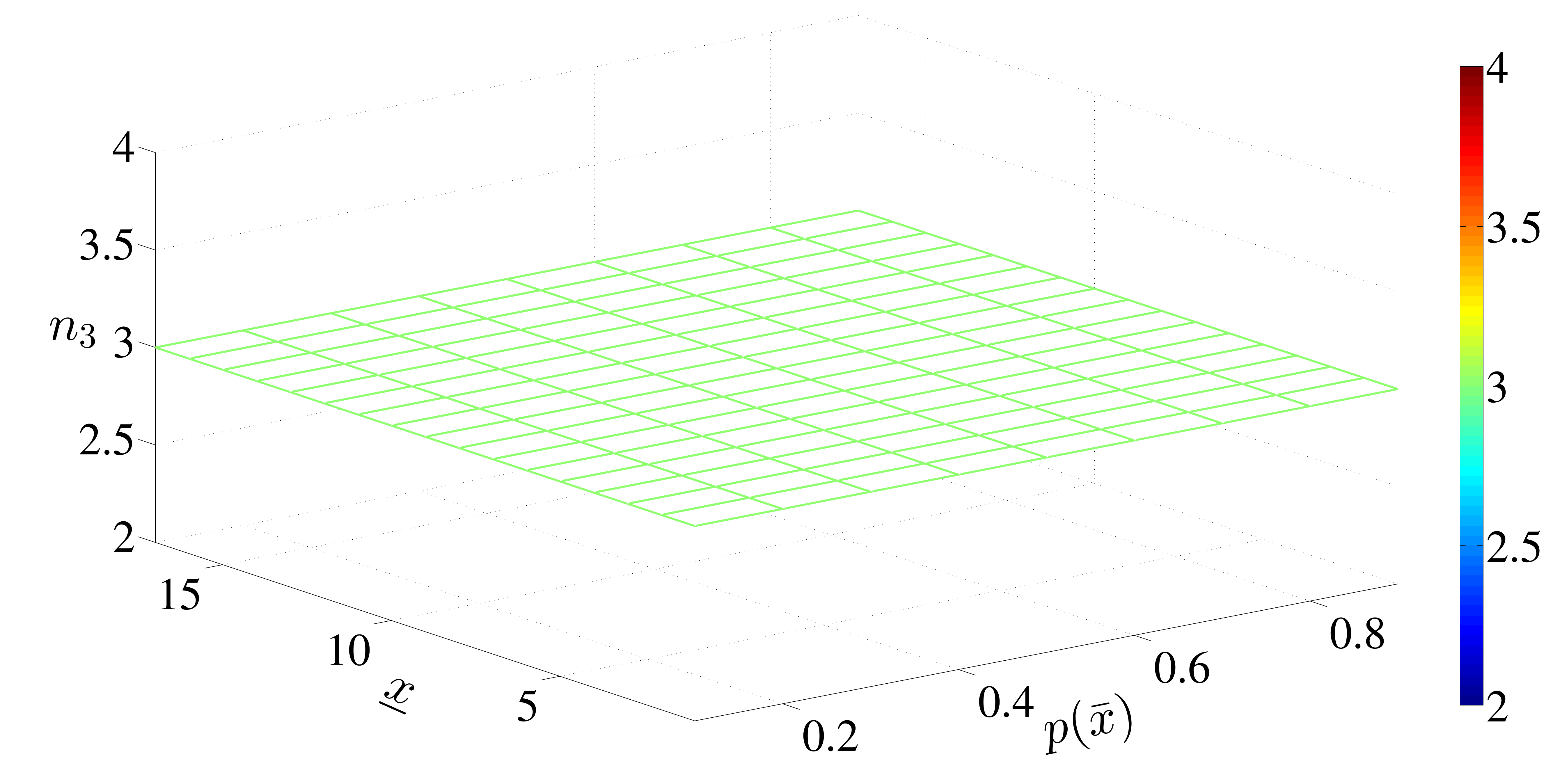}
		\end{minipage} \\ \hline
		$\Delta x=18$ & \begin{minipage}{0.28\textwidth}
			\includegraphics[width=5cm]{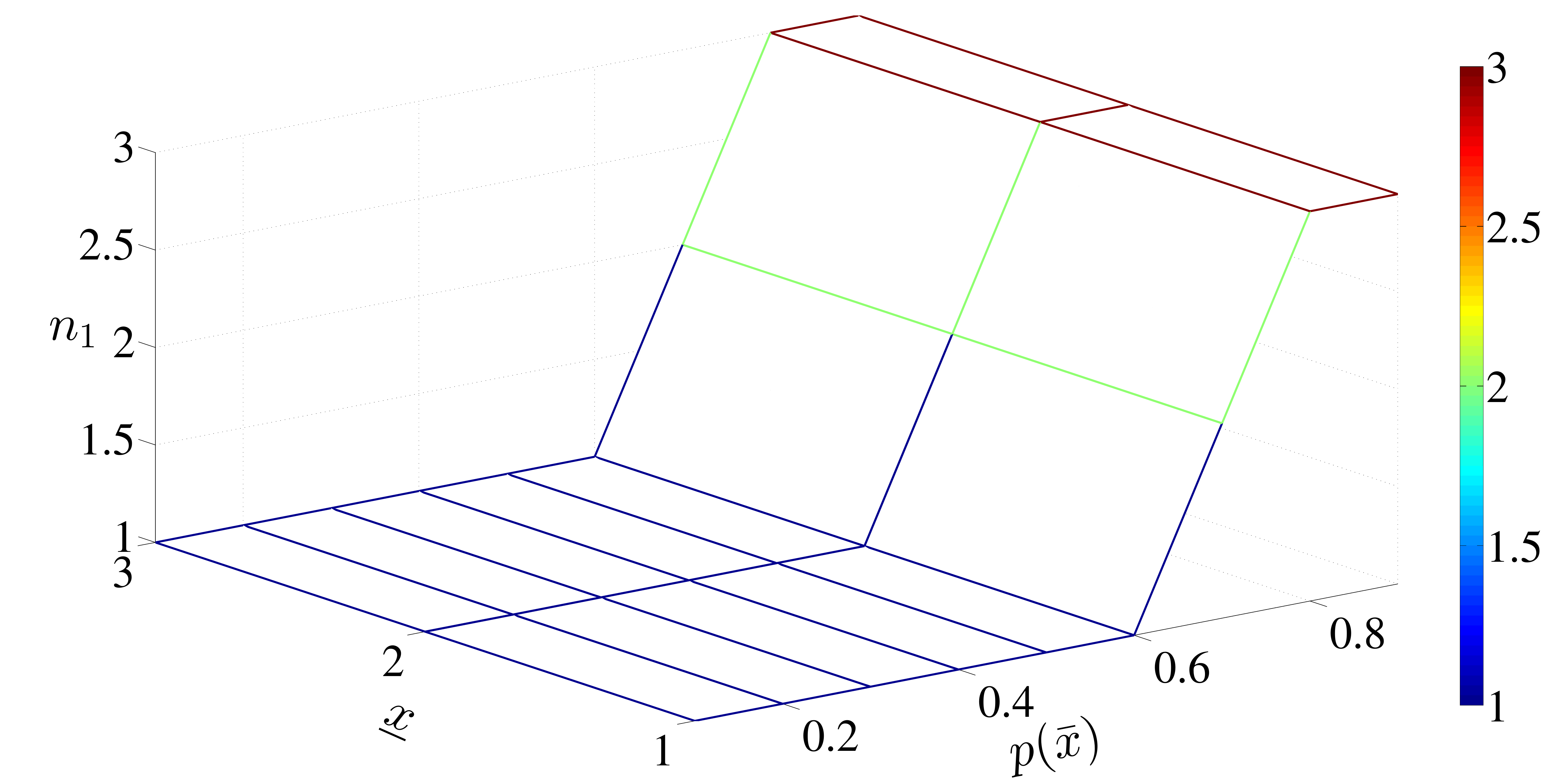}
		\end{minipage} & \begin{minipage}{0.28\textwidth}
			\includegraphics[width=5cm]{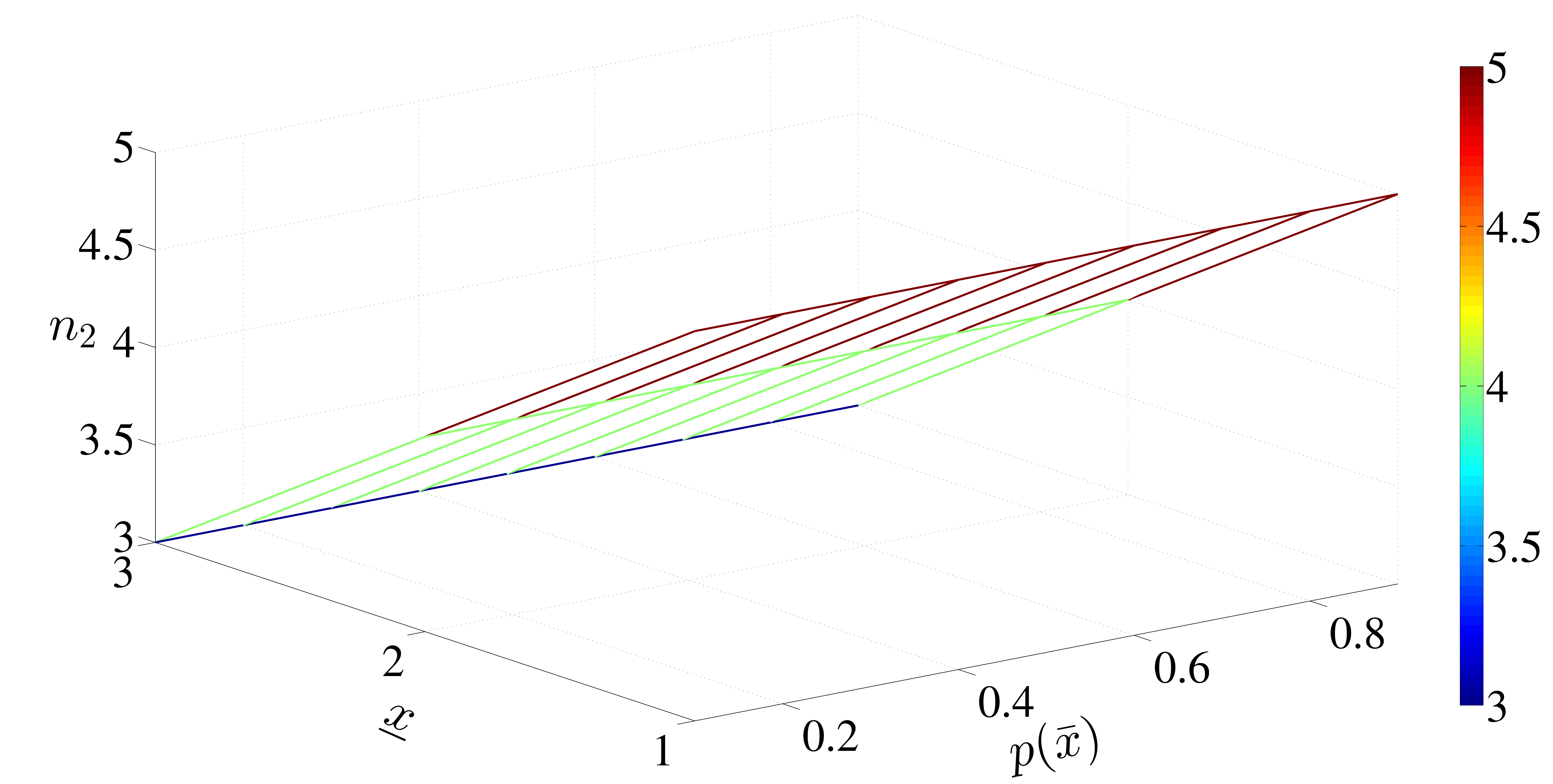}
		\end{minipage} & \begin{minipage}{0.28\textwidth}
			\includegraphics[width=5cm]{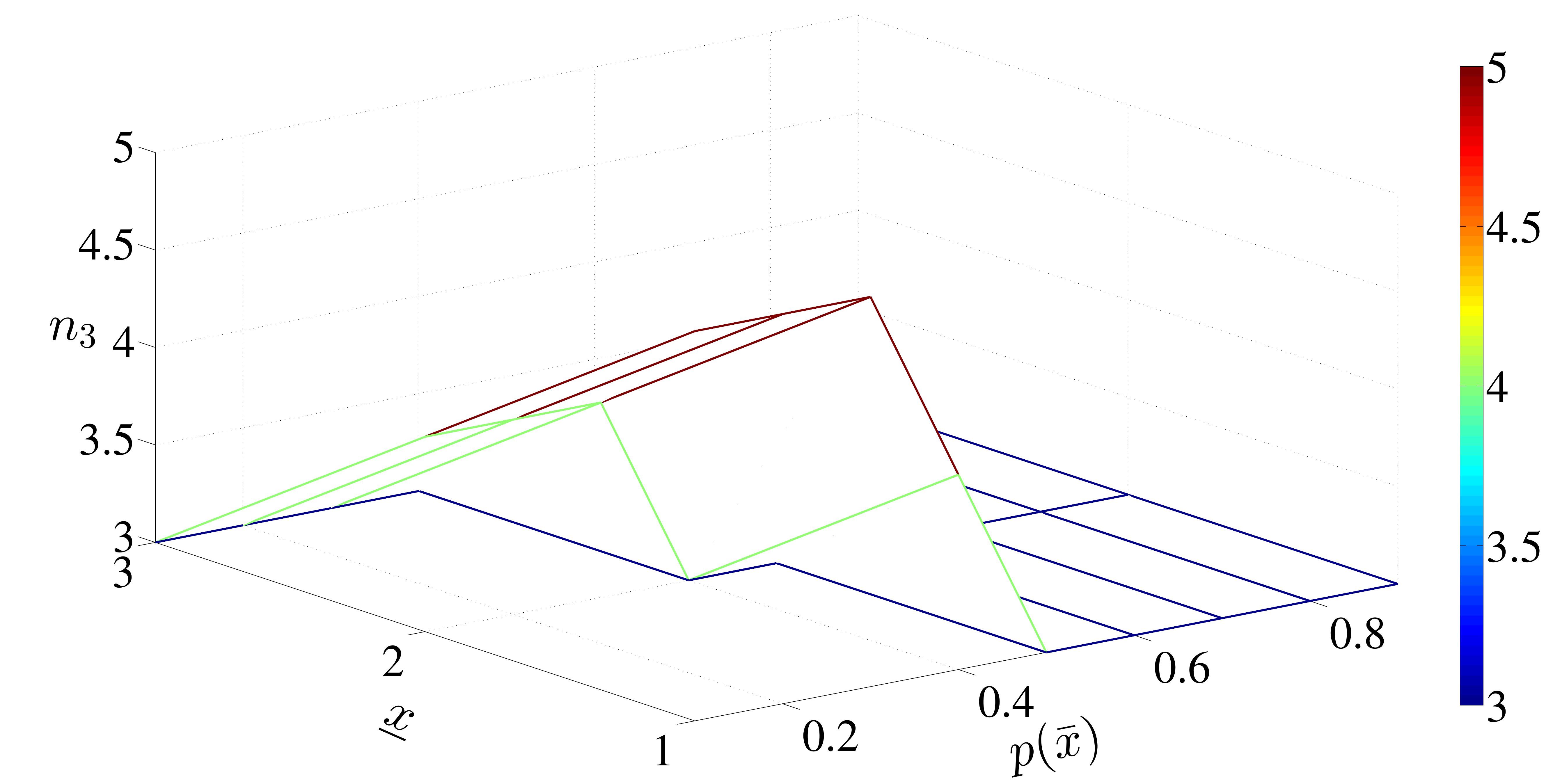}
		\end{minipage} \\ \hline
		$p(\bar{x})=0.2$ & \begin{minipage}{0.28\textwidth}
			\includegraphics[width=5cm]{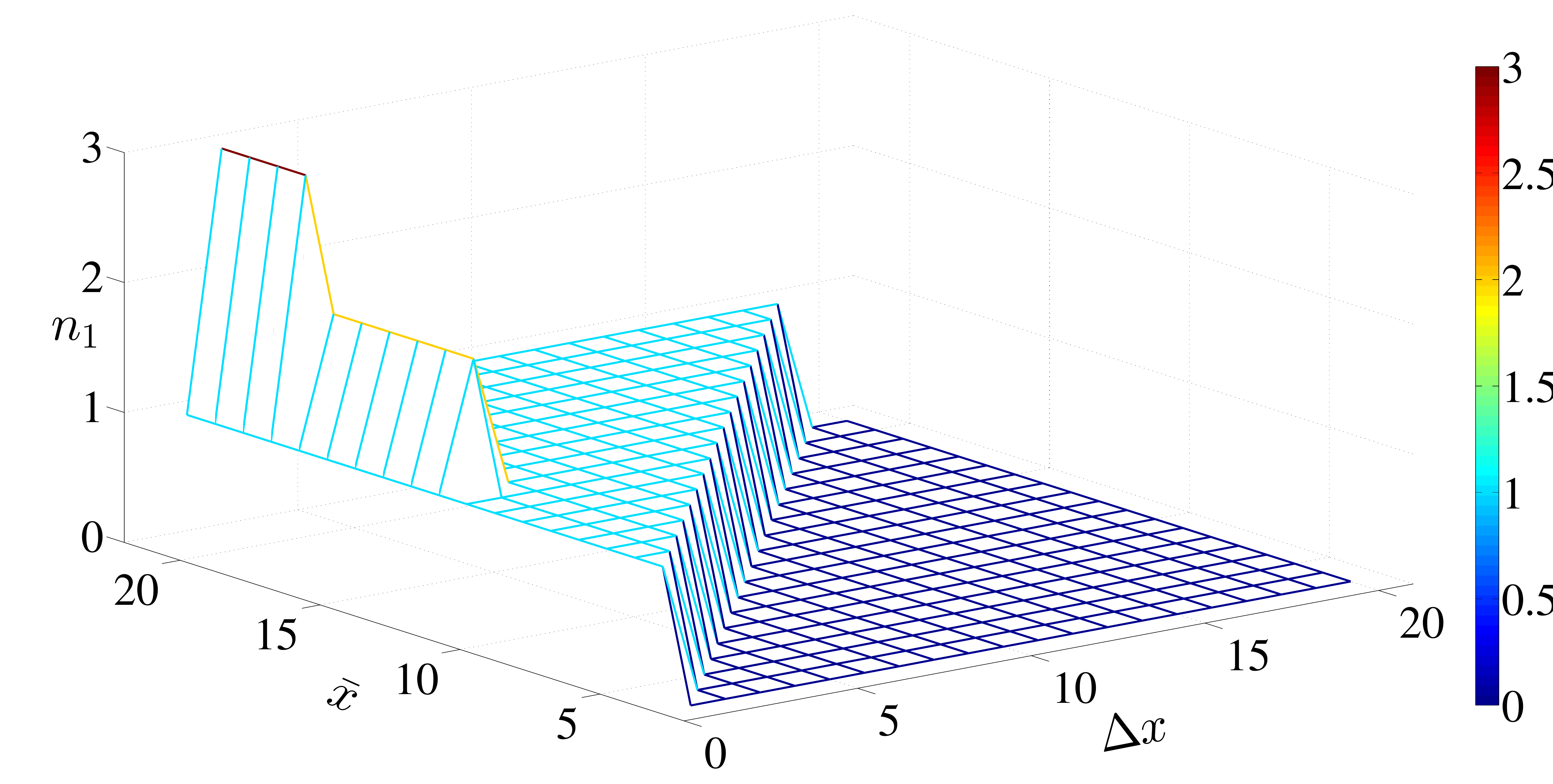}
		\end{minipage} & \begin{minipage}{0.28\textwidth}
			\includegraphics[width=5cm]{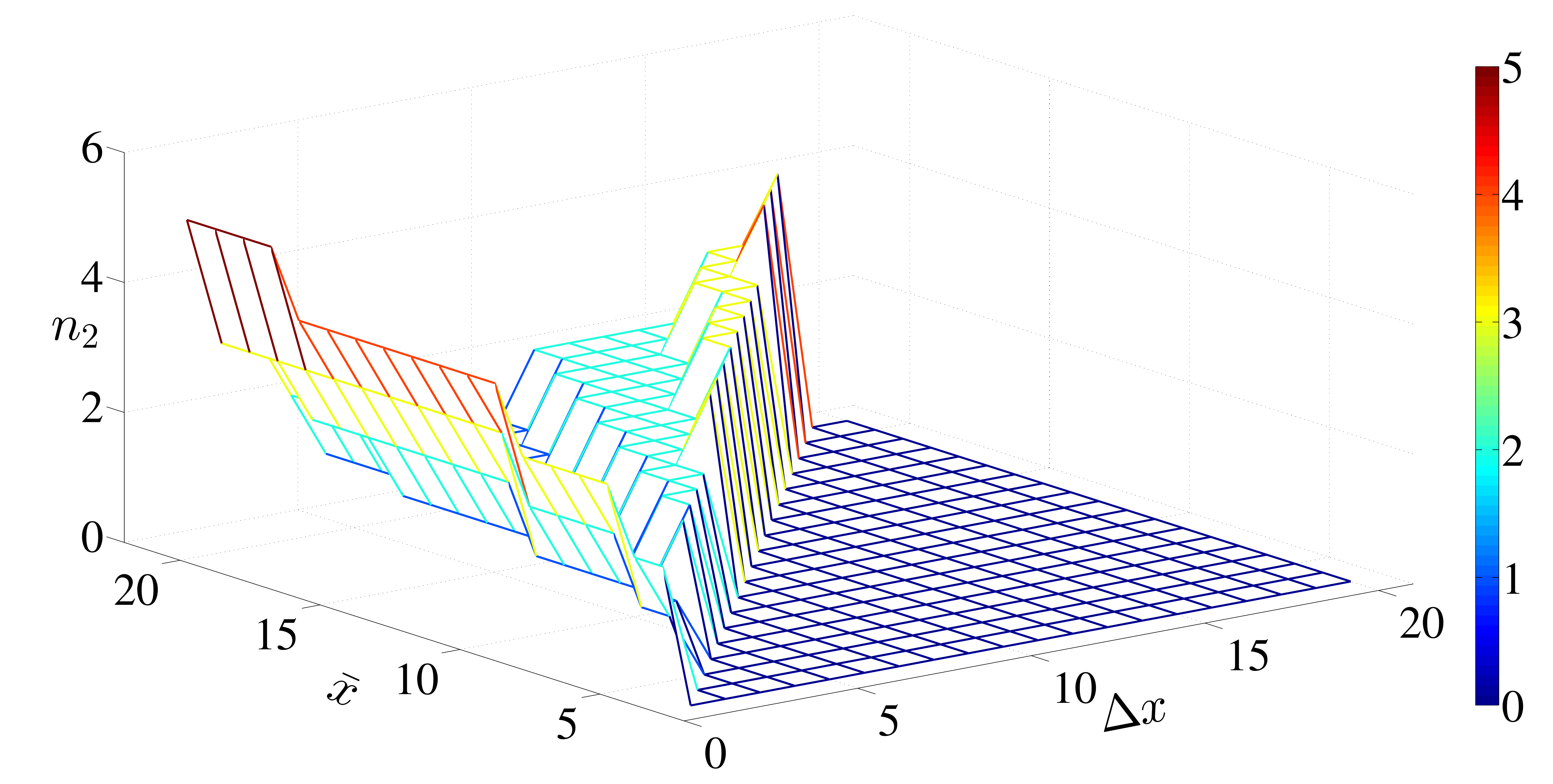}
		\end{minipage} & \begin{minipage}{0.28\textwidth}
			\includegraphics[width=5cm]{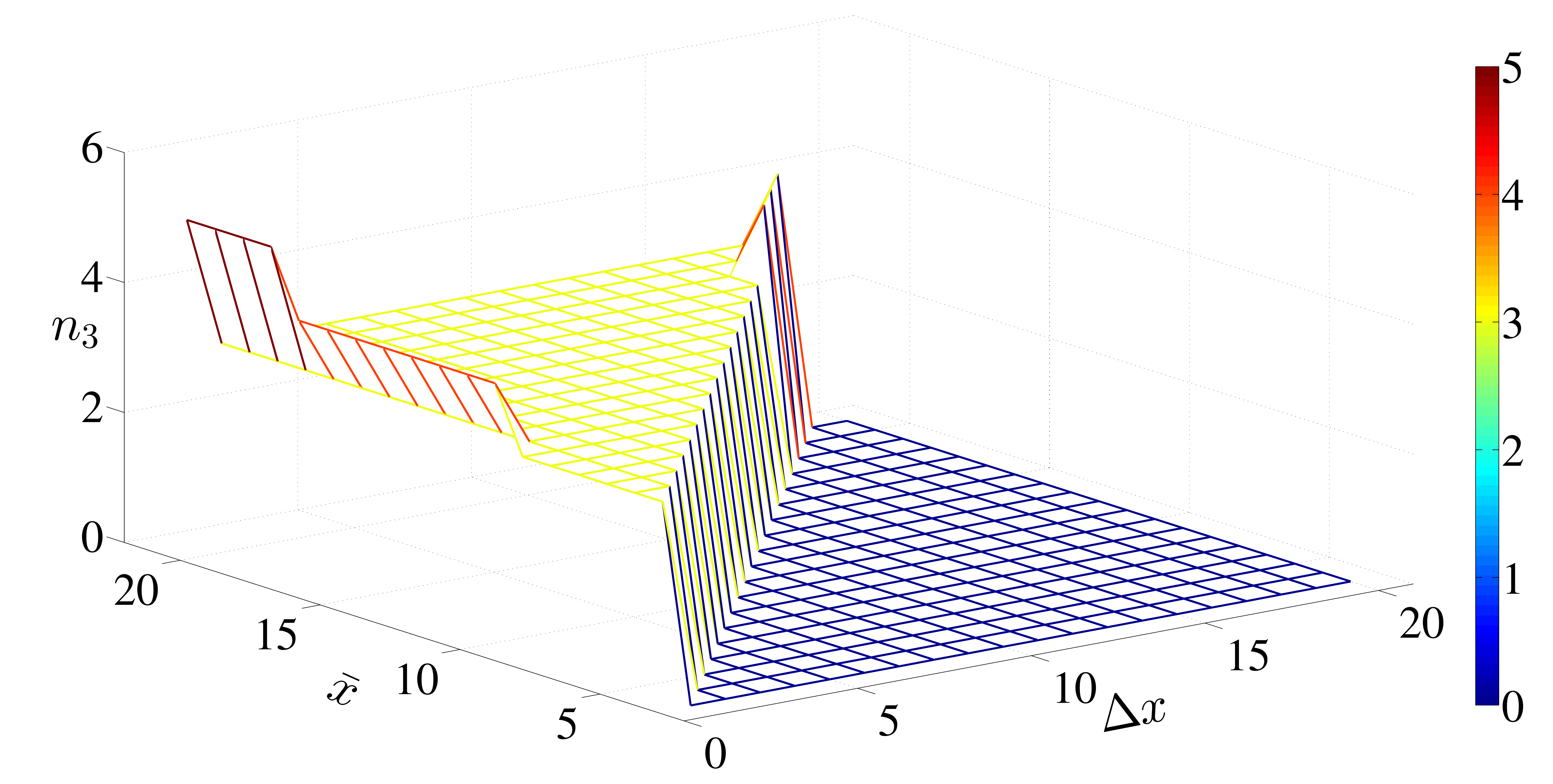}
		\end{minipage} \\ \hline
		$p(\bar{x})=0.9$ & \begin{minipage}{0.28\textwidth}
			\includegraphics[width=5cm]{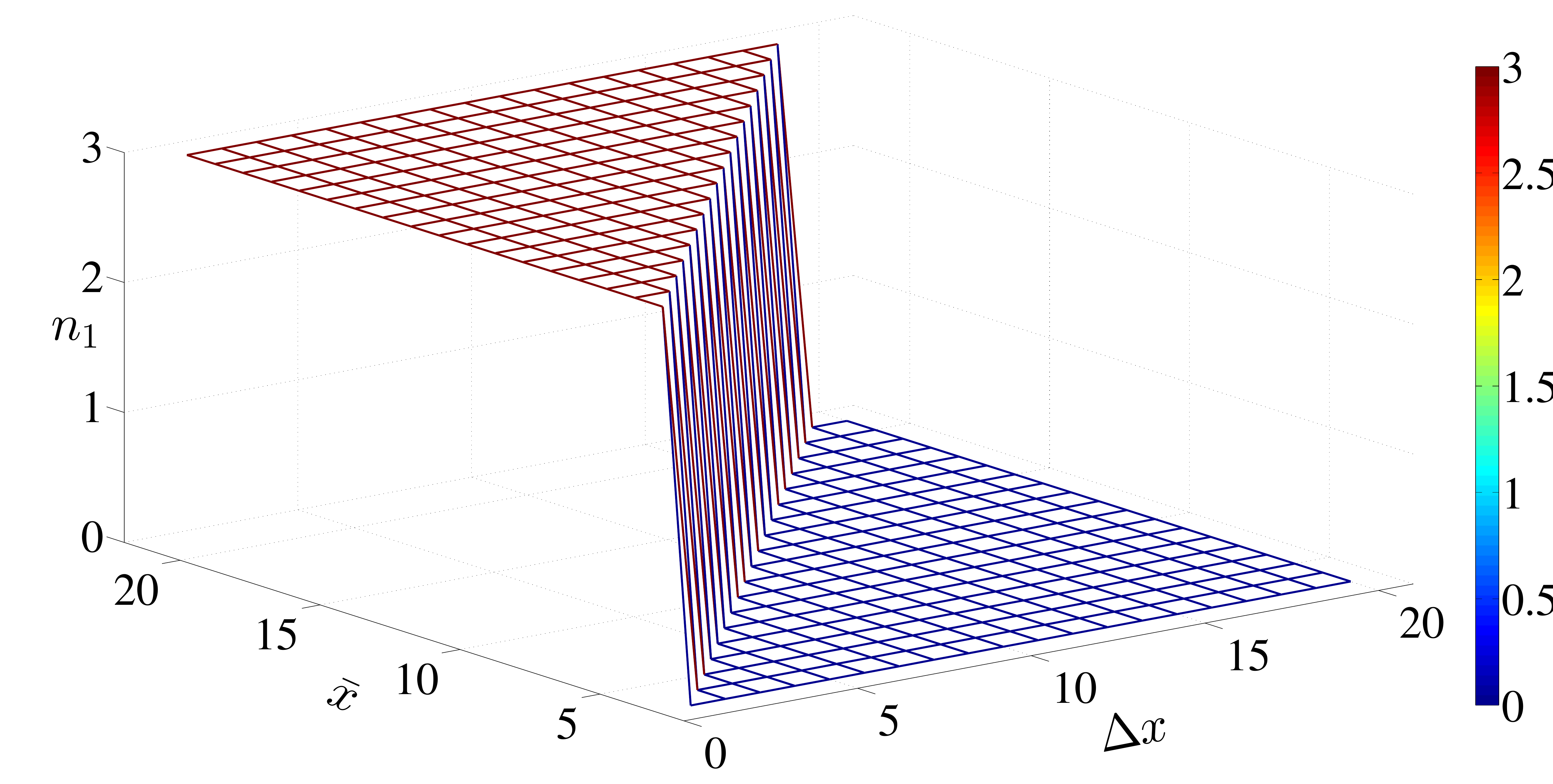}
		\end{minipage} & \begin{minipage}{0.28\textwidth}
			\includegraphics[width=5cm]{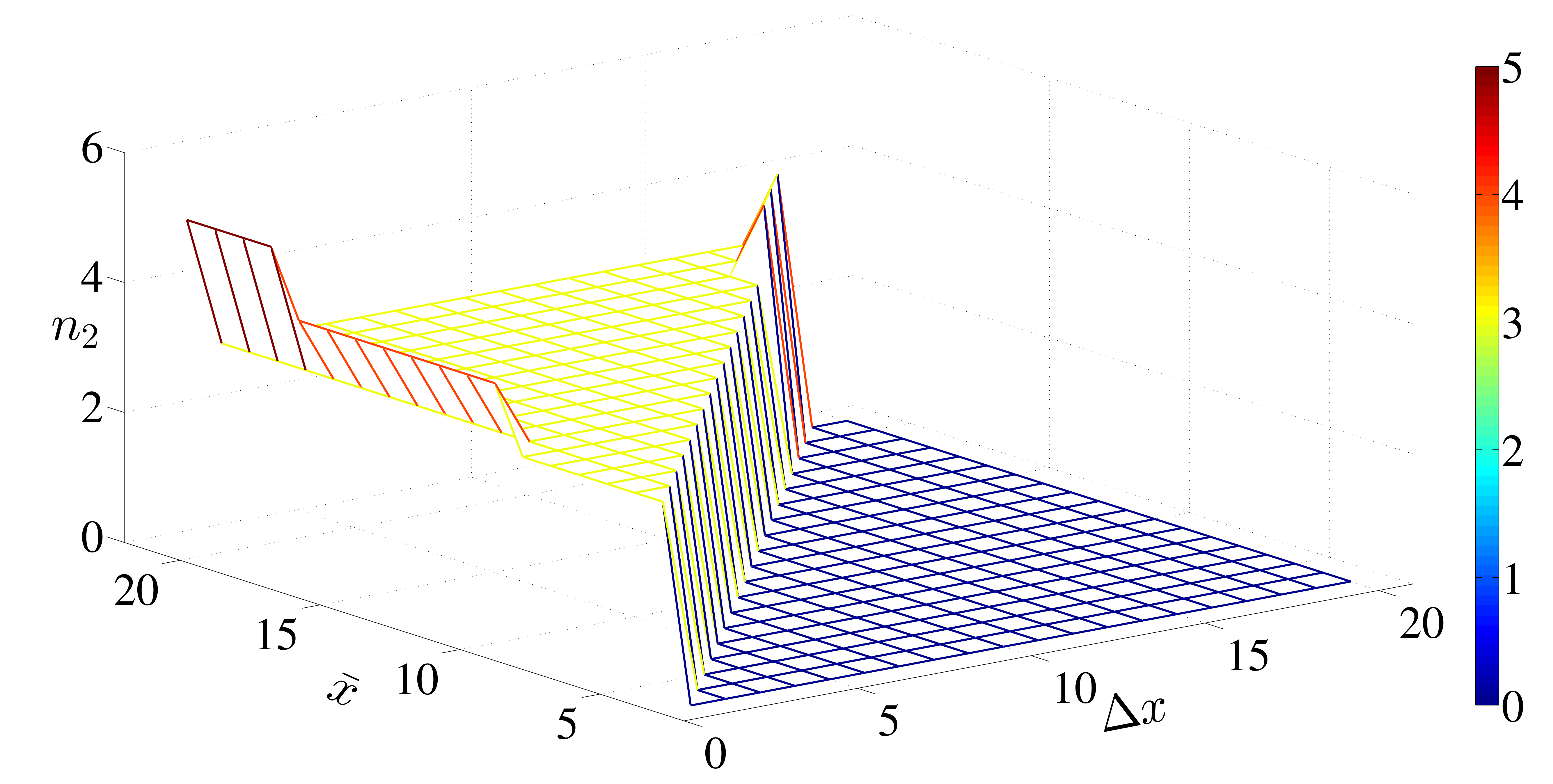}
		\end{minipage} & \begin{minipage}{0.28\textwidth}
			\includegraphics[width=5cm]{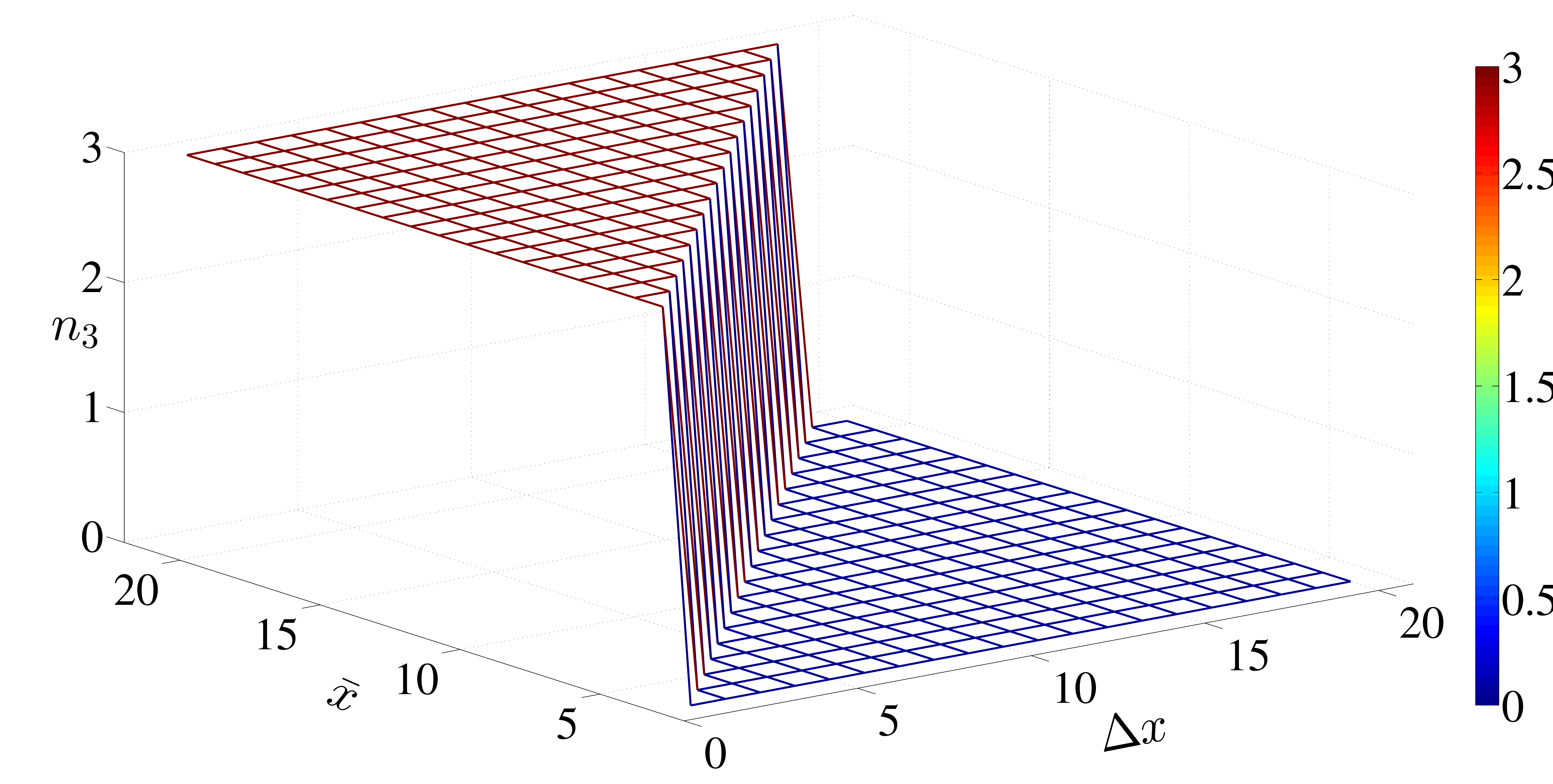}
		\end{minipage} \\ \hline
		$\bar{x}=4$ & \begin{minipage}{0.28\textwidth}
			\includegraphics[width=5cm]{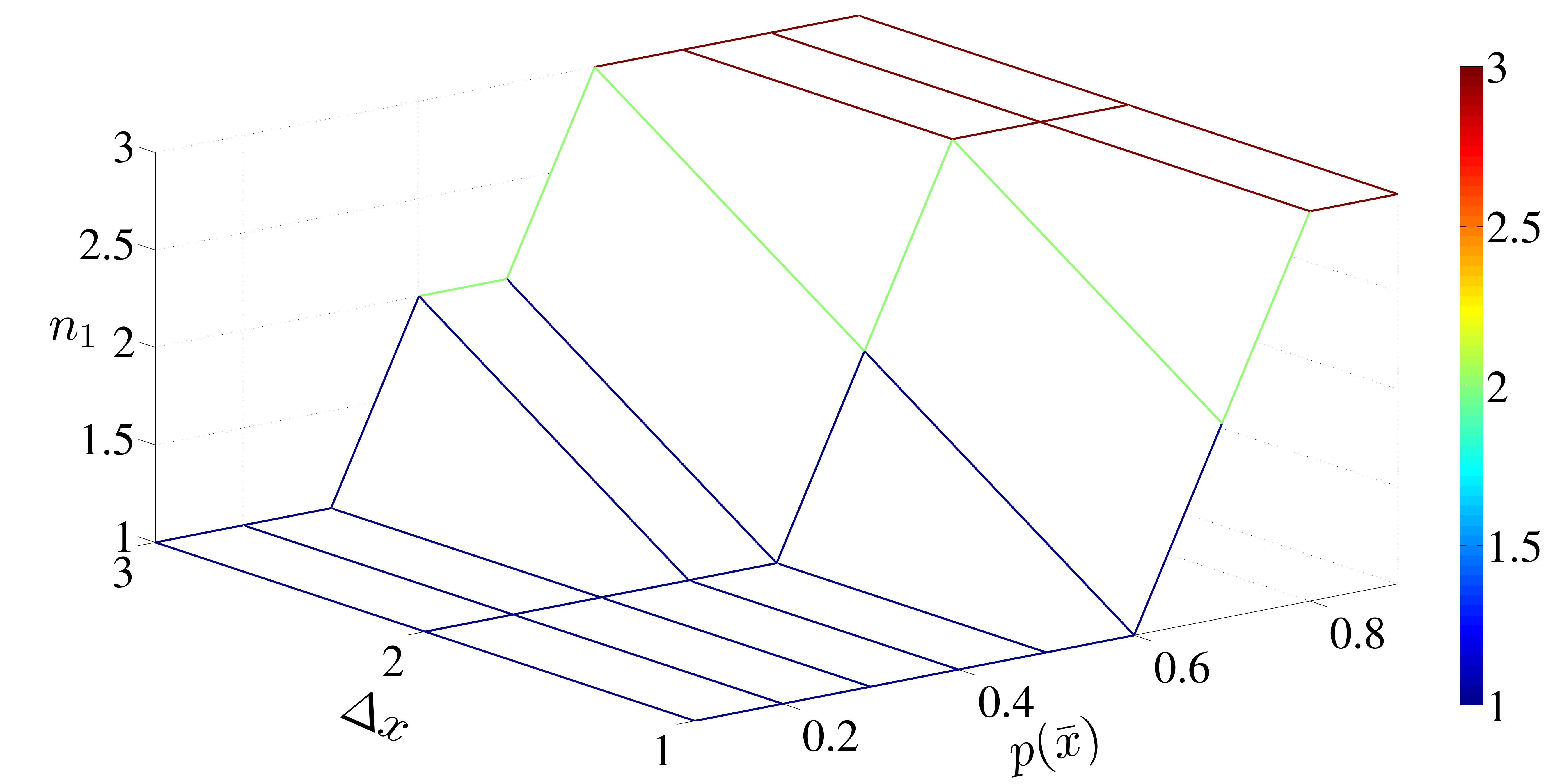}
		\end{minipage} & \begin{minipage}{0.28\textwidth}
			\includegraphics[width=5cm]{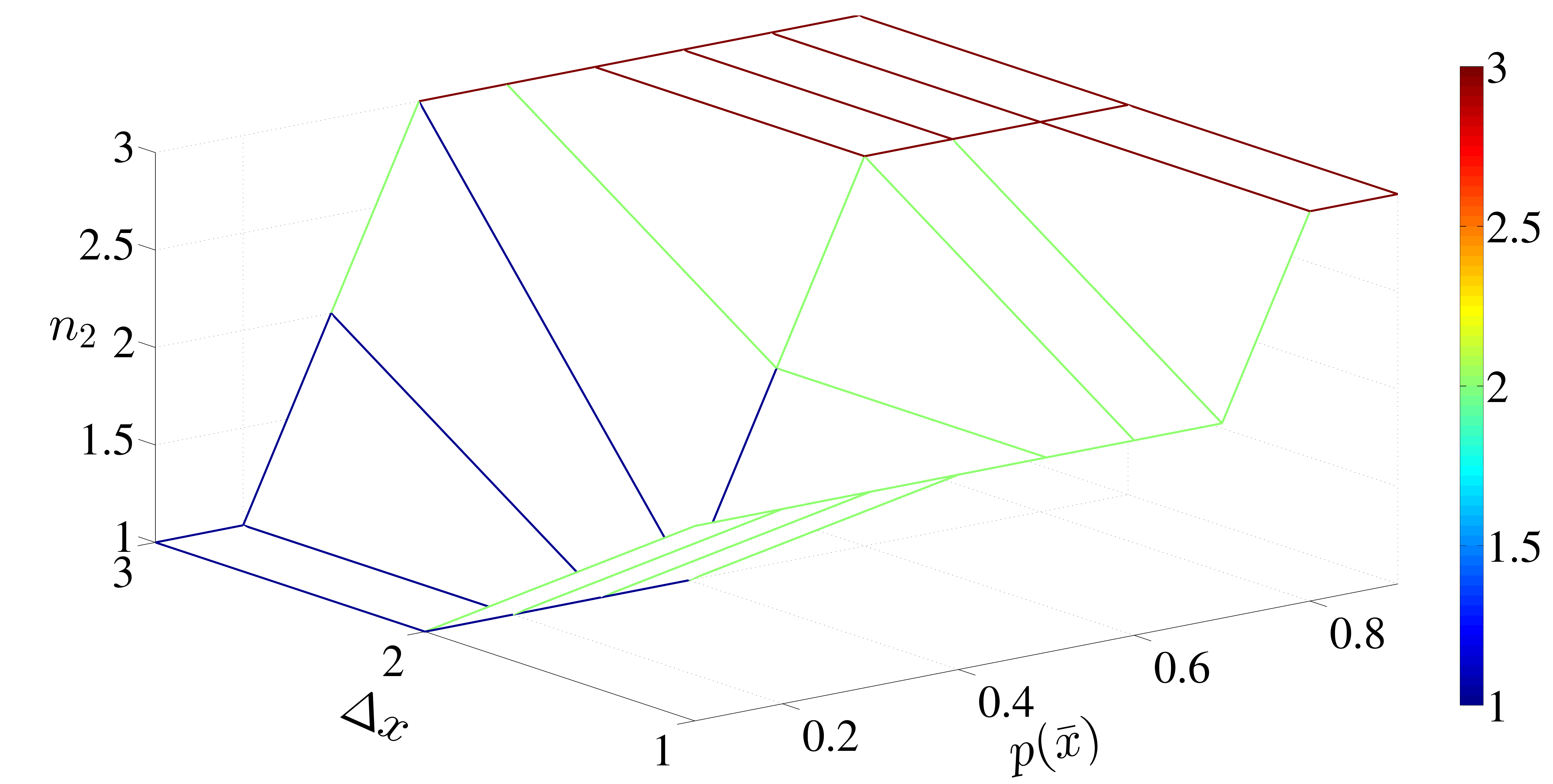}
		\end{minipage} & \begin{minipage}{0.28\textwidth}
			\includegraphics[width=5cm]{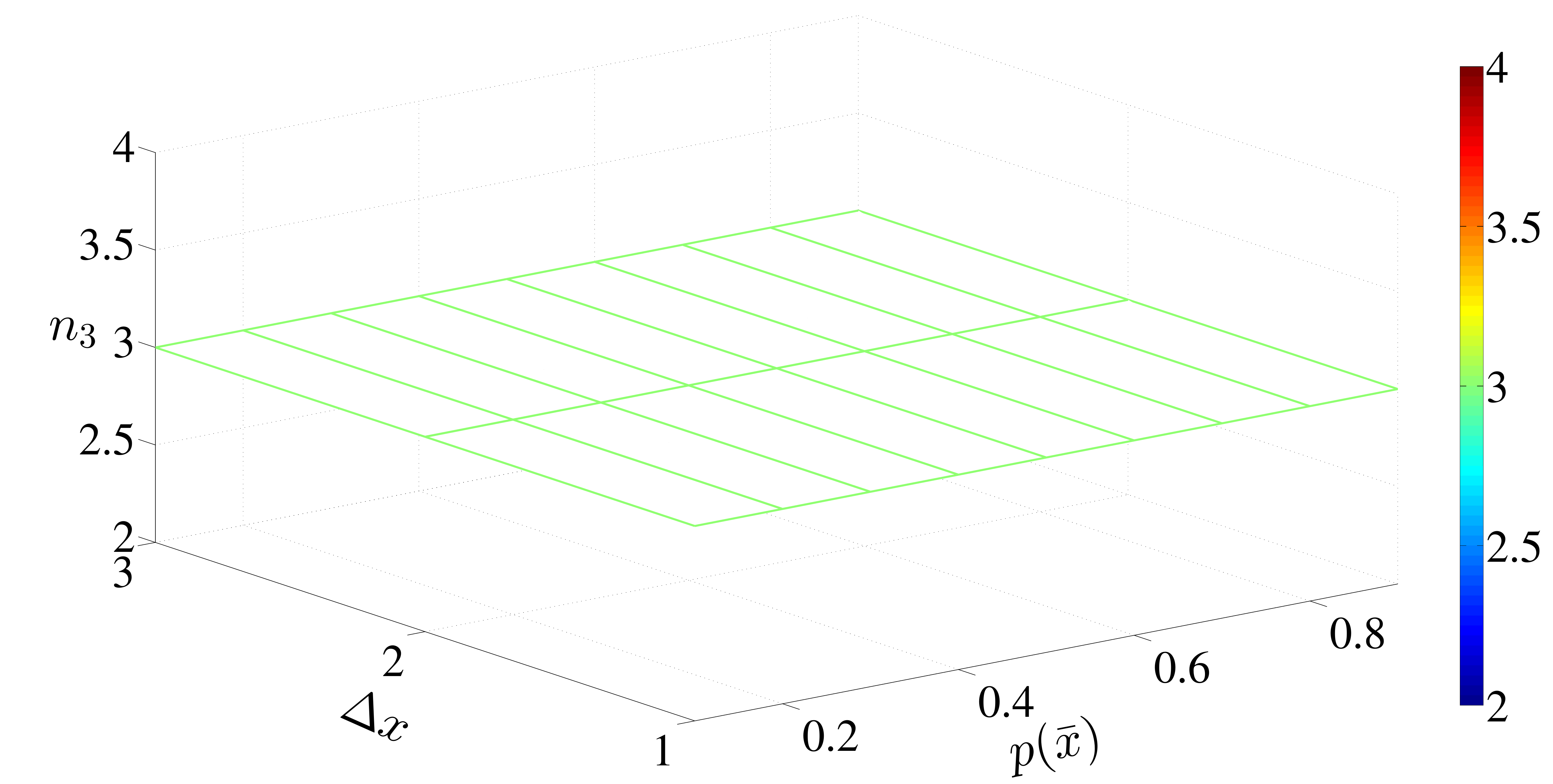}
		\end{minipage} \\ \hline
		$\bar{x}=18$ & \begin{minipage}{0.28\textwidth}
			\includegraphics[width=5cm]{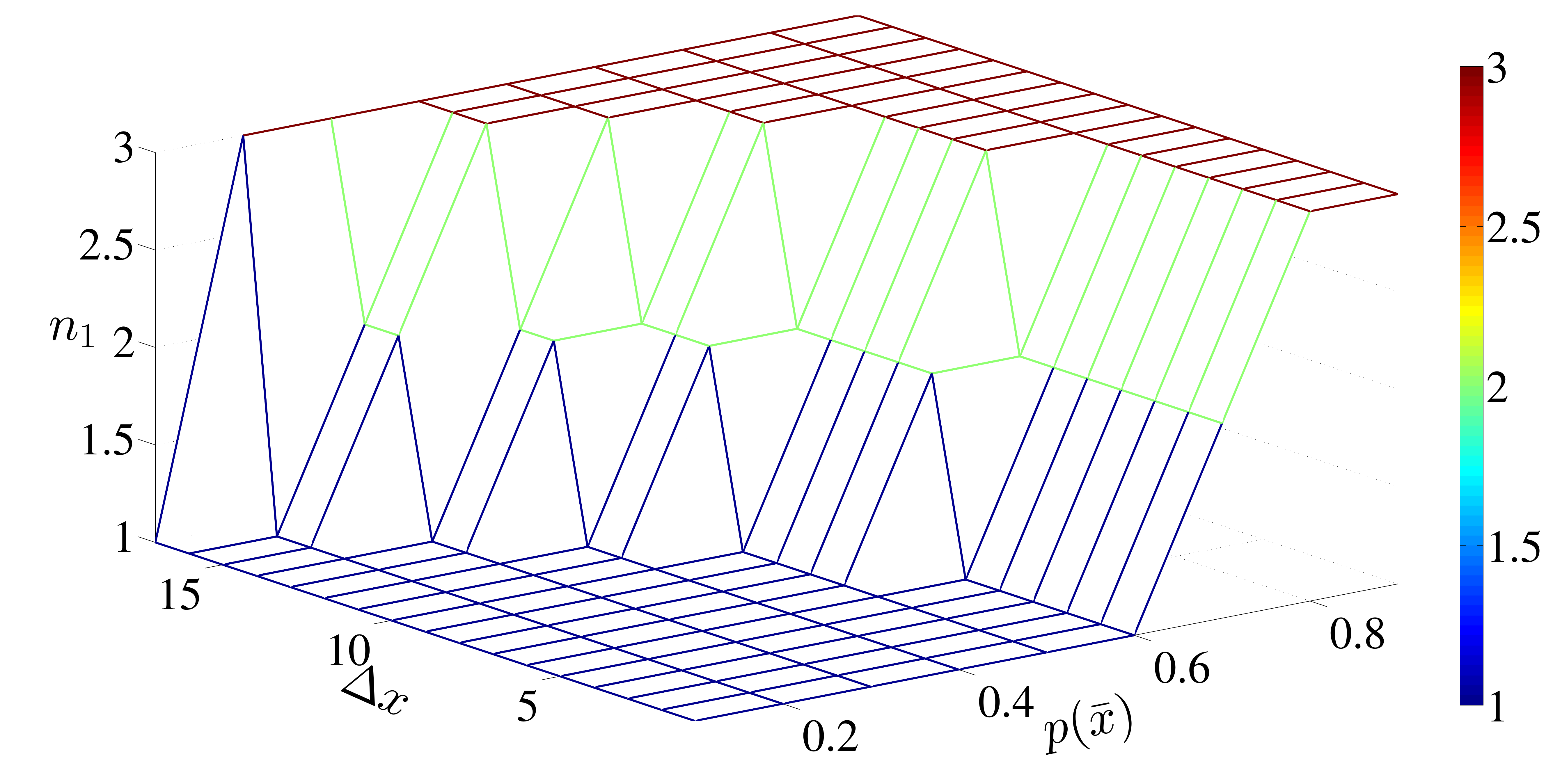}
		\end{minipage} & \begin{minipage}{0.28\textwidth}
			\includegraphics[width=5cm]{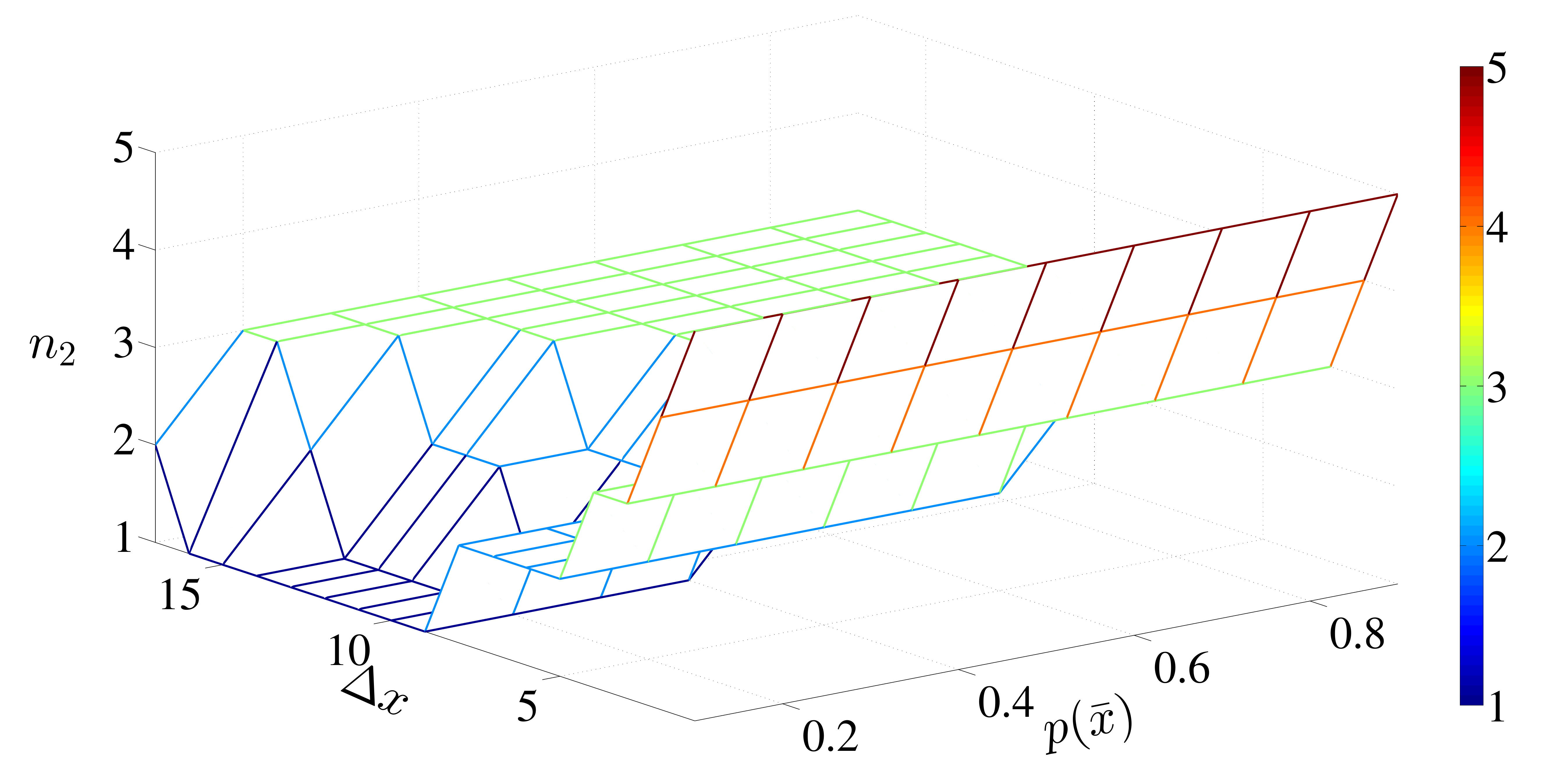}
		\end{minipage} & \begin{minipage}{0.28\textwidth}
			\includegraphics[width=5cm]{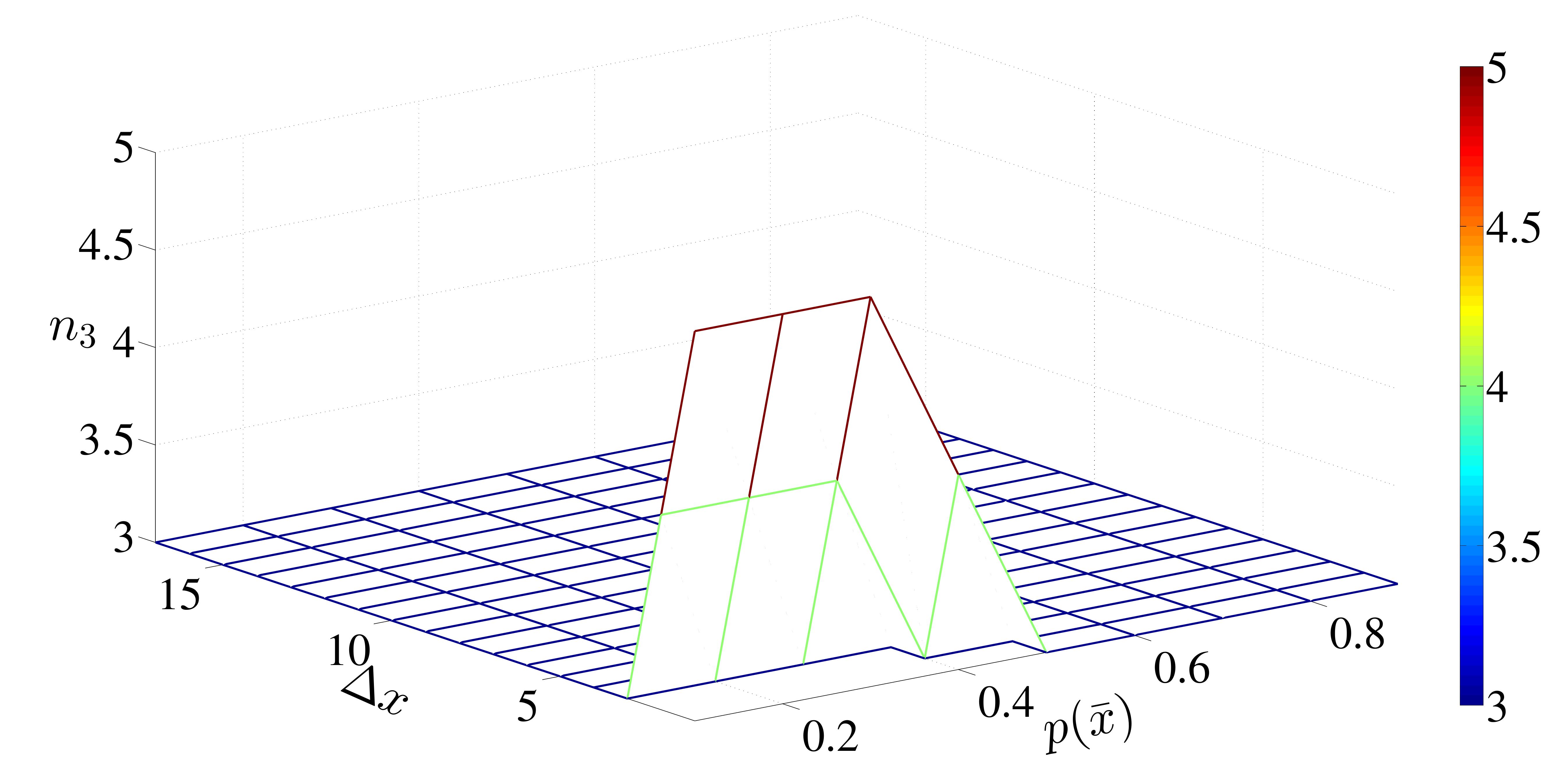}
		\end{minipage} \\ \hline
	\end{tabular}
\end{table*}

It is worth noting that since $t_1$, $t_2$ and $t_3$ in the experiments present no obvious trend\footnote{The underlying reason is the limitation of positive integers for $n_1$, $n_2$ and $n_3$. And thus when $n_i~(i\in \{1,2,3\})$ changes in the integer dimension, the originally continuous $t_i$ will come to abrupt change so we can't find obvious changing trend of $t_1$, $t_2$ and $t_3$.}, we do not analyze them in detail.

\subsection{Three workers}
Similar to the situation of two workers, we utilize the following set of functions to calculate \eqref{eq:19} to \eqref{eq:22} and \eqref{eq:25}:
$f(n)=2^n-1$, $R_1=3\bar{x}n_1^2$, $R_2=(2\bar{x}+\underline{x})n_2^2$, $R_3=(\bar{x}+2\underline{x})n_3^2$, $R_4=3\underline{x}n_4^2$.
We still limit $n_i(i\in\{1,2,3,4\})$ to positive integers and do linear programming to solve the constrained optimization problem.

In Fig. \ref{chart6-7}, we compare the different trends of requestor's utility $\mathbb{U}_r$ changing with $p(\bar{x})$ when the numbers of workers are two and three, respectively, with two different parameter settings of $\bar{x}$ and $\underline{x}$. 
Both figures show that the requestor's utility is positively related to $p(\bar{x})\in [0.1,0.9]$.
It is worth noting that according to the experimental results, the requestor's utility does not increase with the increasing number of workers. To be specific, when $p=0.4$ in Fig. \ref{figure13}, the requestor's  optimal utility of the two-worker case is higher than that of the three-worker case, while when $p=0.9$, three workers is much more better than two workers. 
This is because more workers incur more constraints in the optimization problem, which leads to that fewer workers sometimes can bring more utility to the requestor.

Note that the changing trends of utility with $p(\bar{x})$, $\bar{x}$, $\underline{x}$ is similar to that of the two-worker situation, so we exclude these results due to the limitation of page length.
\begin{figure}
	\centering
	\subfigure[$\bar{x}=80,\underline{x}=20$.]{
		\includegraphics[width=0.21\textwidth]{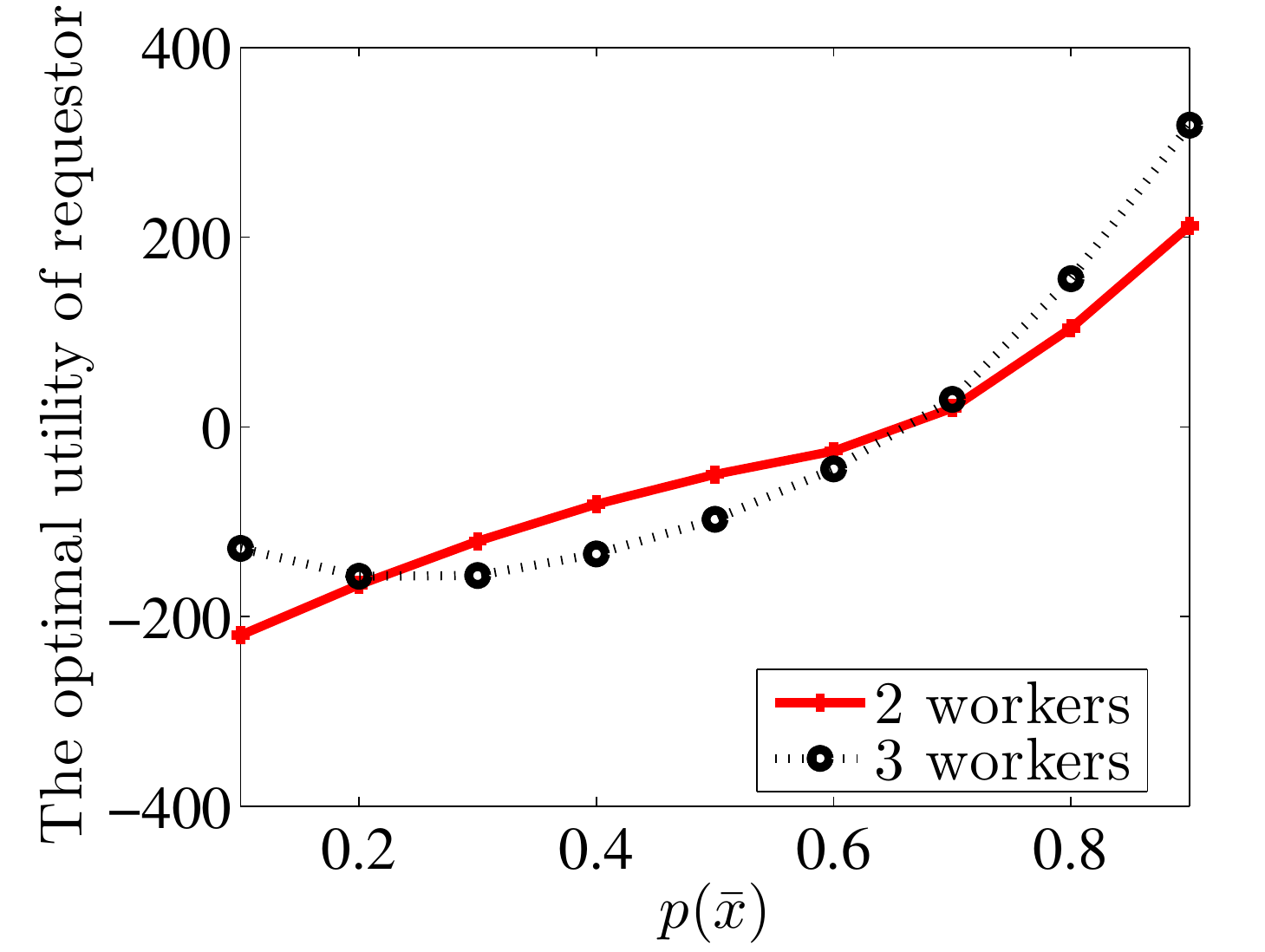}
		\label{figure13}
	}
	\quad
	\subfigure[$\bar{x}=50,\underline{x}=10$.]{
		\includegraphics[width=0.21\textwidth]{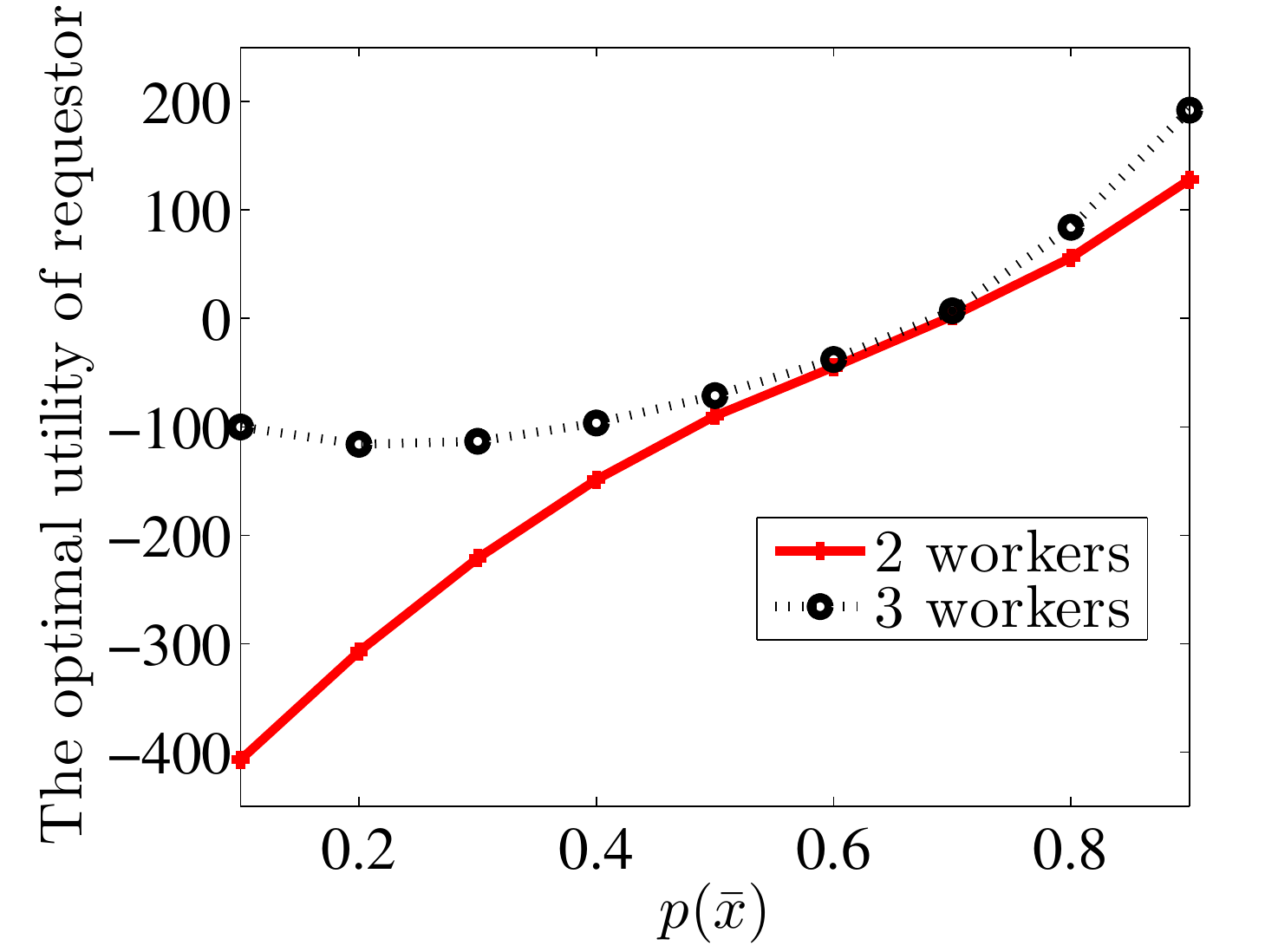}
		\label{figure14}
	}
	\caption{Optimal utility comparison between two workers and three workers.}
	\label{chart6-7}
\end{figure}

\section{Conclusion}\label{sec:conclusion}
In this paper, we propose a misreport- and
collusion-proof crowdsourcing mechanism, guiding workers to truthfully report the quality of submitted tasks without collusion by leveraging pricing and
task allocation.  Extensive simulation results verify the effectiveness of the
proposed mechanism, based on which we can obtain three counterintuitive findings:
1) a high-quality worker may pretend to be a low-quality
one; 2) the rise of task quality from high-quality workers may not
result in the increased utility of the requestor; 3)  the utility of the requestor may not get improved with the
increasing number of workers. In addition, we can also draw the following conclusions: 1) with the increase of $p(\bar{x})$, $n_1$ and $n_2$ increase  while $n_3$ decreases; 2) the utility of the requestor goes up with the increase of  $p(\bar{x})$. Moreover,
the higher the $\Delta x$, the larger the increment of her utility.

\bibliographystyle{IEEEtran}
\bibliography{reference}

\begin{thebibliography}{10}
\providecommand{\url}[1]{#1}
\csname url@samestyle\endcsname
\providecommand{\newblock}{\relax}
\providecommand{\bibinfo}[2]{#2}
\providecommand{\BIBentrySTDinterwordspacing}{\spaceskip=0pt\relax}
\providecommand{\BIBentryALTinterwordstretchfactor}{4}
\providecommand{\BIBentryALTinterwordspacing}{\spaceskip=\fontdimen2\font plus
\BIBentryALTinterwordstretchfactor\fontdimen3\font minus
  \fontdimen4\font\relax}
\providecommand{\BIBforeignlanguage}[2]{{%
\expandafter\ifx\csname l@#1\endcsname\relax
\typeout{** WARNING: IEEEtran.bst: No hyphenation pattern has been}%
\typeout{** loaded for the language `#1'. Using the pattern for}%
\typeout{** the default language instead.}%
\else
\language=\csname l@#1\endcsname
\fi
#2}}
\providecommand{\BIBdecl}{\relax}
\BIBdecl

\bibitem{oleson2011programmatic}
D.~Oleson, A.~Sorokin, G.~Laughlin, V.~Hester, J.~Le, and L.~Biewald,
  ``Programmatic gold: Targeted and scalable quality assurance in
  crowdsourcing,'' 2011.

\bibitem{wu2017photo}
Y.~Wu, Y.~Wang, and G.~Cao, ``Photo crowdsourcing for area coverage in resource
  constrained environments,'' pp. 1--9, 2017.

\bibitem{xu2015revealing}
A.~Xu, X.~Feng, and Y.~Tian, ``Revealing, characterizing, and detecting
  crowdsourcing spammers: A case study in community q\&a,'' pp. 2533--2541,
  2015.

\bibitem{jin2017leveraging}
Y.~Jin, M.~Carman, D.~Kim, and L.~Xie, ``Leveraging side information to improve
  label quality control in crowd-sourcing,'' pp. 1--10, 2017.

\bibitem{qiu2017dynamic}
C.~Qiu, A.~C. Squicciarini, S.~M. Rajtmajer, and J.~Caverlee, ``Dynamic
  contract design for heterogenous workers in crowdsourcing for quality
  control,'' pp. 1168--1177, 2017.

\bibitem{hu2019quality}
Q.~Hu, S.~Wang, P.~Ma, X.~Cheng, W.~Lv, and R.~Bie, ``Quality control in
  crowdsourcing using sequential zero-determinant strategies,'' \emph{IEEE
  Transactions on Knowledge and Data Engineering}, pp. 1--11, 2019.

\bibitem{han2016crowdsourcing}
S.~Han, P.~Dai, P.~Paritosh, and D.~Huynh, ``Crowdsourcing human annotation on
  web page structure: Infrastructure design and behavior-based quality
  control,'' \emph{ACM Transactions on Intelligent Systems and Technology
  (TIST)}, vol.~7, no.~4, p.~56, 2016.

\bibitem{he2015high}
Z.~He, J.~Cao, and X.~Liu, ``High quality participant recruitment in
  vehicle-based crowdsourcing using predictable mobility,'' pp. 2542--2550,
  2015.

\bibitem{tarable2015importance}
A.~Tarable, A.~Nordio, E.~Leonardi, and M.~A. Marsan, ``The importance of being
  earnest in crowdsourcing systems,'' pp. 2821--2829, 2015.

\bibitem{borgers2015introduction}
T.~B{\"o}rgers, \emph{An introduction to the theory of mechanism design}.\hskip
  1em plus 0.5em minus 0.4em\relax Oxford University Press, USA, 2015.

\bibitem{whitehill2009whose}
J.~Whitehill, T.-f. Wu, J.~Bergsma, J.~R. Movellan, and P.~L. Ruvolo, ``Whose
  vote should count more: Optimal integration of labels from labelers of
  unknown expertise,'' pp. 2035--2043, 2009.

\bibitem{peng2018data}
D.~Peng, F.~Wu, and G.~Chen, ``Data quality guided incentive mechanism design
  for crowdsensing,'' \emph{IEEE transactions on mobile computing}, vol.~17,
  no.~2, pp. 307--319, 2018.

\end{thebibliography}

\clearpage
\begin{IEEEbiography}[{\includegraphics[width=1in,height=1.25in,clip,keepaspectratio]{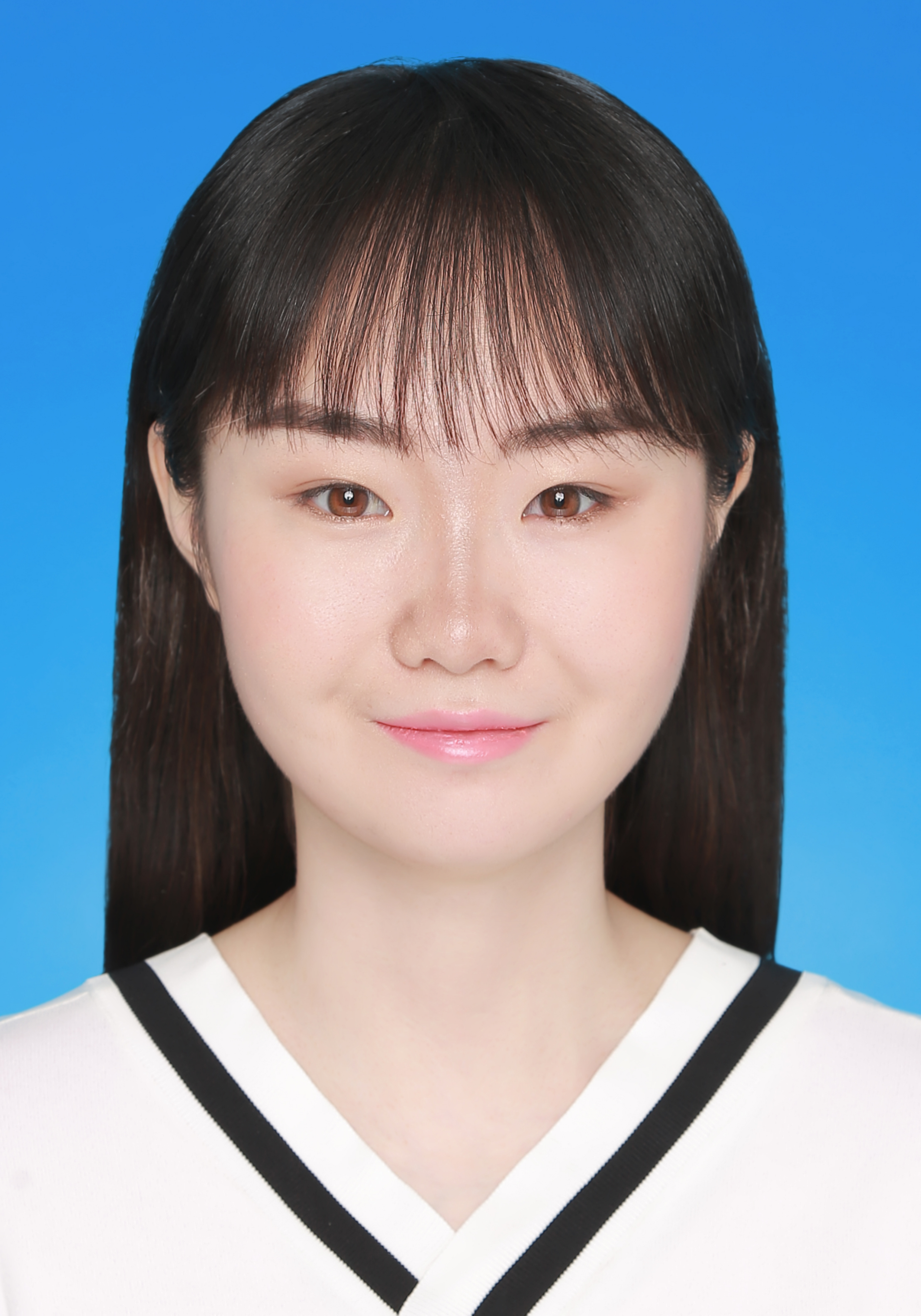}}]{Kun Li}
	received the B.S. degree in 2017 from  School of Artificial Intelligence, Beijing Normal University. Now, she is a graduate student in School of Artificial Intelligence, Beijing Normal University. She joined the Center for Big Data Mining \& Knowledge Engineering for research work on September 2017. Her research interests include crowdsourcing and game theory.
\end{IEEEbiography}
\begin{IEEEbiography}[{\includegraphics[width=1in,height=1.25in,clip,keepaspectratio]{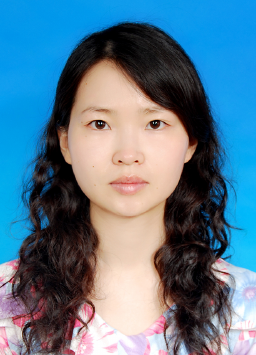}}]{Shengling Wang}
	is a full professor in the School of Artificial Intelligence, Beijing Normal University. She received her Ph.D. in 2008 from Xi'an Jiaotong University. After that, she did her postdoctoral research in the Department of Computer Science and Technology, Tsinghua University. Then she worked as an assistant and associate professor from 2010 to 2013 in the Institute of Computing Technology of the Chinese Academy of Sciences. Her research interests include mobile/wireless networks, game theory, crowdsourcing.
\end{IEEEbiography}
\begin{IEEEbiography}[{\includegraphics[width=1in,height=1.25in,clip,keepaspectratio]{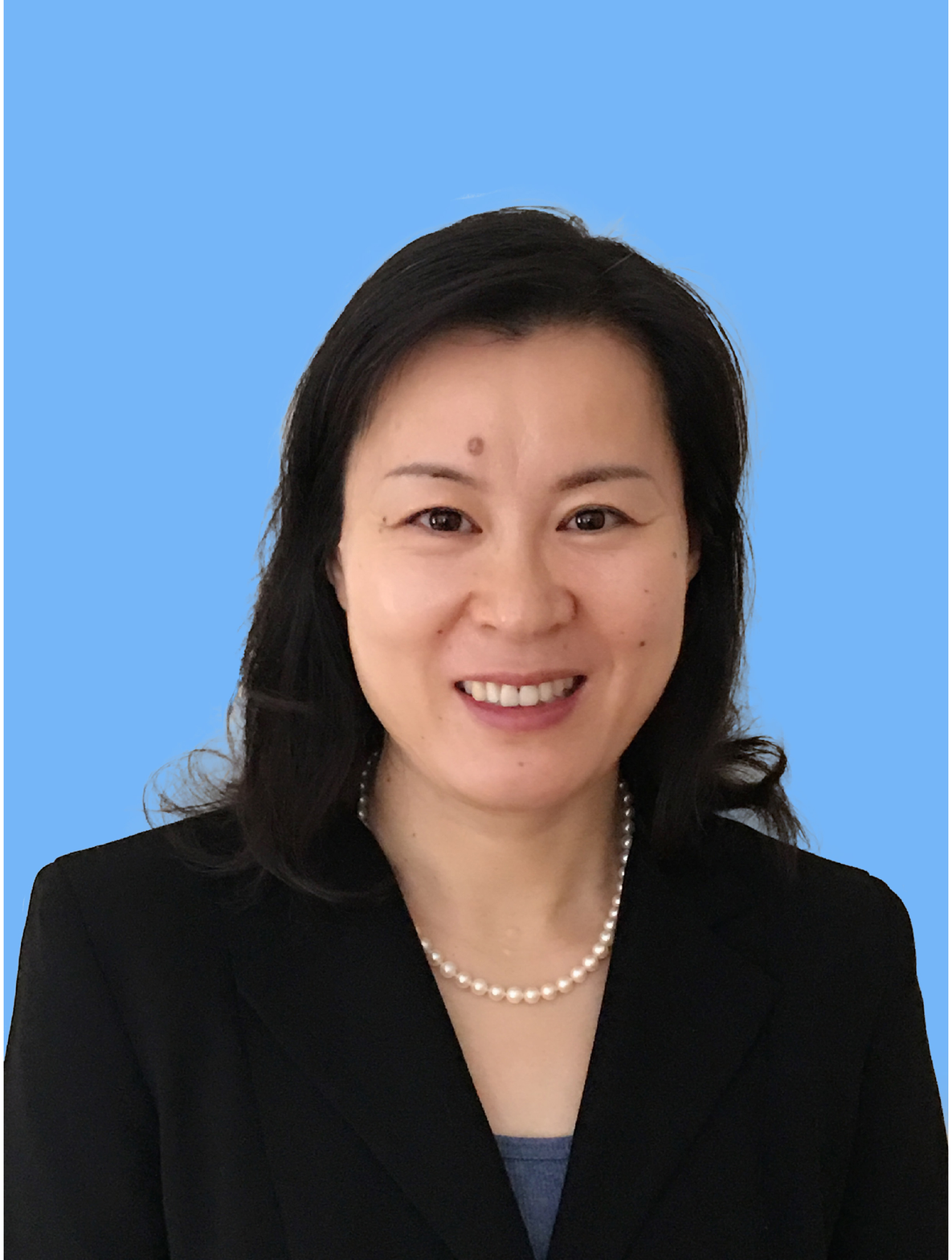}}]{Xiuzhen Cheng}
	received her M.S. and Ph.D. degrees in computer science from the University of Minnesota Twin Cities in 2000 and 2002, respectively. She is a professor in the School of Computer Science and Technology, Shandong University. Her current research interests focus on privacy-aware computing, wireless and mobile security, dynamic spectrum access, mobile handset networking systems (mobile health and safety), cognitive radio networks, and algorithm design and analysis.
	She has served on the Editorial Boards of several technical publications and the Technical Program Committees of various professional conferences/workshops. She has also chaired several international conferences. She worked as a program director for the U.S. National Science Foundation (NSF) from April to October 2006 (full time), and from April 2008 to May 2010 (part time). She published more than 170 peerreviewed papers.
\end{IEEEbiography}
\begin{IEEEbiography}[{\includegraphics[width=1in,height=1.25in,clip,keepaspectratio]{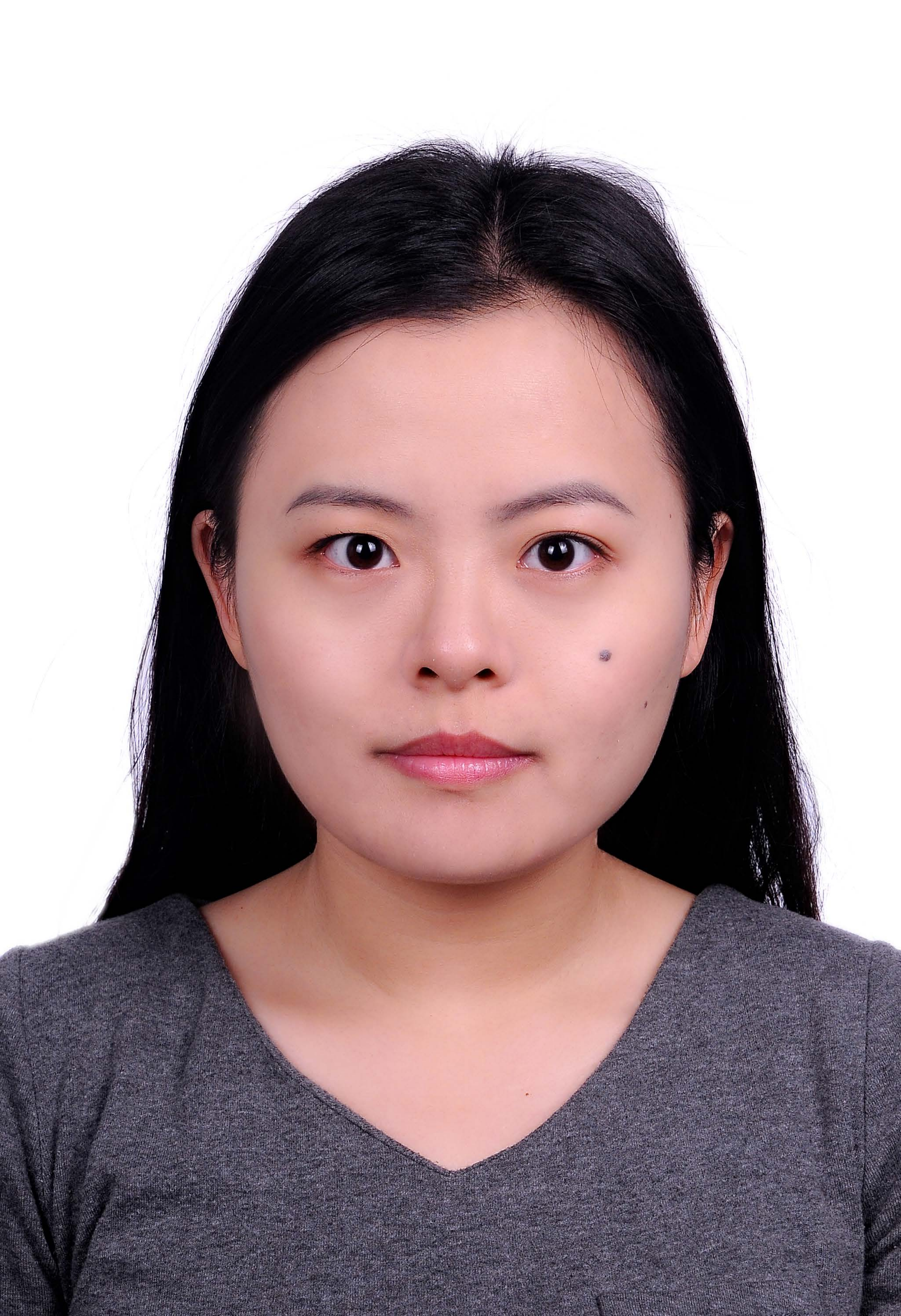}}]{Qin Hu} received her Ph.D. degree in Computer Science from the George Washington University in 2019. She is currently an Assistant Professor in the department of Computer and Information Science, Indiana University - Purdue University Indianapolis. Her research interests include wireless and mobile security, crowdsourcing/crowdsensing and blockchain.
\end{IEEEbiography}
\vfill

\end{document}